\documentclass[singlecolumn, journal]{IEEEtran}
\usepackage[cmex10]{amsmath}
\usepackage{graphicx} 
\usepackage[numbers,sort&compress]{natbib}
\usepackage{epsfig}
\usepackage{amsmath}
\usepackage{amssymb}
\usepackage{footnote}
\usepackage{array}
\usepackage{verbatim, animate}
\usepackage{comment, amsmath, amsfonts, amsthm}
\usepackage[mathscr]{euscript}
\usepackage{epsfig}
\usepackage{amsfonts,amssymb,latexsym,amsmath, amsxtra,verbatim}
\usepackage[all]{xy}
\usepackage{psfrag}
\usepackage{xcolor}

\usepackage{xspace}
\usepackage{bbm}
\input{mathlig}

\newcommand{\muspace}{\mspace{1mu}}

\DeclareRobustCommand{\scond}{\mathchoice{\muspace\vert\muspace}{\vert}{\vert}{\vert}}
\mathlig{|}{\scond}

\DeclareRobustCommand{\discint}{\mathchoice{\mspace{-1.5mu}:\mspace{-1.5mu}}{\mspace{-1.5mu}:\mspace{-1.5mu}}{:}{:}}
\mathlig{::}{\discint}

\def\co{\mathop{\rm co}\nolimits}%

\newcommand{\Pc}{\mathcal{P}}
\newcommand{\Qc}{\mathcal{Q}}

\newcommand{\Uc}{\mspace{1.5mu}\mathcal{U}}
\newcommand{\Vc}{\mathcal{V}}

\newcommand{\Xc}{\mathcal{X}}

\newcommand{\Qcal}{\mathcal{Q}}

\newcommand{\Ucal}{\mathcal{U}}

\newcommand{\Xcal}{\mathcal{X}}

\renewcommand{\Pr}{\mathscr{P}}
\newcommand{\Rr}{\mathscr{R}}

\def\a{\alpha}
\def\b{\beta}

\def\l{\lambda}

\DeclareMathOperator\E{\textsf{E}}
\let\P\relax
\DeclareMathOperator\P{\textsf{P}}

\newcommand{\Bern}{\mathrm{Bern}}

\def\textiid{i.i.d.\@\xspace}
\newcommand\iid{\ifmmode\text{ i.i.d. } \else \textiid \fi}

\def\mathllap{\mathpalette\mathllapinternal}
\def\mathllapinternal#1#2{%
  \llap{$\mathsurround=0pt#1{#2}$}}

\def\clap#1{\hbox to 0pt{\hss#1\hss}}
\def\mathclap{\mathpalette\mathclapinternal}
\def\mathclapinternal#1#2{%
  \clap{$\mathsurround=0pt#1{#2}$}}

\let\oldstackrel\stackrel
\renewcommand{\stackrel}[2]{\oldstackrel{\mathclap{#1}}{#2}}

\renewcommand{\hbar}{h\mathllap{\overline{\vphantom{h}\hphantom{\rule{4.6pt}{0pt}}}\mspace{0.77mu}}}

\catcode`~=11 
\newcommand{\urltilde}{\kern -.06em\lower -.06em\hbox{~}\kern .02em}
\catcode`~=13 

\usepackage{lineno}
\usepackage{auto-pst-pdf}
\usepackage{pst-pdf}
\usepackage[font=footnotesize]{caption}
\usepackage{subcaption}
\usepackage{multirow}

\newtheorem{theorem}{\textbf{Theorem}}

\newtheorem{lemma}{\textbf{Lemma}}

\newtheorem{definition}{\textbf{Definition}}

\newtheorem{remark}{\textbf{Remark}}

\newcommand{\highlight }{\color{blue}}

\newcommand{\CC}{\mathfrak{C}}

\onecolumn
\IEEEoverridecommandlockouts
\makeatother

\begin{document}
\title{Superposition Coding is Almost Always Optimal for the Poisson Broadcast Channel\\
\vspace{-5pt}}

\author{Hyeji Kim\IEEEauthorrefmark{1}, Benjamin Nachman\IEEEauthorrefmark{2} and Abbas El Gamal\IEEEauthorrefmark{1}

\thanks{\IEEEauthorrefmark{1} Hyeji Kim and Abbas El Gamal are with the Department of Electrical Engineering, Stanford University (email: hyejikim@stanford.edu and abbas@ee.stanford.edu).}\thanks{\IEEEauthorrefmark{2} Benjamin Nachman is with the Department of Physics, Stanford University (email: nachman@stanford.edu).}
\thanks{ Hyeji Kim is supported by the Alma M. Collins Stanford Graduate Fellowship. Benjamin Nachman is supported by the NSF Graduate Research Fellowship under Grant No. DGE-4747 and by the Stanford Graduate Fellowship. This paper was presented in part at \emph{Proc. IEEE Int. Symp. Inf. Theory, Hong Kong, 2015}.}%
}

\maketitle


\begin{abstract}
This paper shows that the capacity region of the continuous-time Poisson broadcast channel is achieved via superposition coding for most channel parameter values. Interestingly, the channel in some subset of these parameter values does not belong to any of the existing classes of broadcast channels for which superposition coding is optimal (e.g., degraded, less noisy, more capable). In particular, we introduce the notion of effectively less noisy broadcast channel and show that it implies less noisy but is not in general implied by more capable. For the rest of the channel parameter values, we show that there is a gap between Marton's inner bound and the UV outer bound. 
\end{abstract}
\IEEEpeerreviewmaketitle

\section{Introduction}
The continuous-time Poisson channel is a canonical model of the point to point optical communication channel in the low power regime~\cite{ref3,ref4,ref5}. The capacity of this channel was established using different approaches by Kabanov~\cite{Kabanov78}, Davis~\cite{Davis80}, and Wyner~\cite{Wyner1988a, Wyner1988b}. In particular, Wyner~\cite{Wyner1988a, Wyner1988b} established the capacity using an elementary method in which the capacity is shown to be the the limit of the capacity of a certain memoryless binary channel. Wyner's approach spurred several generalizations to multiple user Poisson channels. In~\cite{Lapidoth--Shamai98} Lapidoth and Shamai established the capacity region of the Poisson multiple-access channel. In~\cite{Lai--Liang--Shamai15}, Lai, Liang, and Shamai studied the Poisson interference channel. In~\cite{Bross--Lapidoth--Wang09}, Bross, Lapidoth, and Shamai studied the Poisson channel with side information at the transmitter. In~\cite{Lapidoth--Telatar--Urbanke2003}, Lapidoth, Telatar, and Erbanke studied the Poisson broadcast channel and established the condition under which the channel is degraded; hence the capacity region is achieved using superposition coding~\cite{Cover1972}. 

In this paper, which is an expanded and a more complete version of~\cite{Kim--Nachman--EG2015}, we show that for the Poisson broadcast channel, superposition coding is optimal much beyond the parameter ranges for which the channel is degraded. We consider the 2-receiver continuous-time Poisson broadcast channel (P-BC) depicted in Figure~\ref{fig1}. The channel input signal $X(t) \in [0,1]$, $t \ge 0$, that is, we assume a \emph{peak power} constraint of 1 on $X(t)$.
 Given $X(t) = x(t)$, the output $Y_i(t)$ is a Poisson process (PP) with instantaneous rate $A_i(x(t)+s_i)$ for $i=1,2$, i.e., for $0 \le w \le w+\tau \le T$,
\[
\P\{Y_i(w+\tau) - Y_i(w) = k | X(t) = x(t),\, t\in [0,T]\} = \frac{\Gamma^k(w,\tau)}{k!}e^{-\Gamma(w,\tau)},\; k \in \mathbb{N}, 
\]
where 
\[
\Gamma(w, \tau) = \int_{w}^{w+\tau} A_i(x(t)+s_i)\ dt.
\]
The parameter $A_i$ is the channel gain for receiver $i=1,2$. The parameter $s_i \ge 0$ is the rate of the input referred dark noise for receiver $i$.

We consider the private message setting in which the sender $X$ wishes to communicate a message $M_1$ to receiver $Y_1$ at rate $R_1$ and a message $M_2$ to receiver $Y_2$ at rate $R_2$, where $M_1$ and $M_2$ are independent and uniformly distributed over $[1:2^{nR_1}]\times [1:2^{nR_2}]$. The results we establish on the optimality of superposition coding can be readily extended to the case with common message~\cite{EG--Kim2011}. We define a $(2^{nR_1},2^{nR_2},n)$ code, achievability, and the capacity region for this setting in the standard way~\cite{EG--Kim2011}. 
\begin{figure}[htpb]
\begin{center}
\psfrag{X}[r]{$X(t)$}
\psfrag{Z1}[b]{$s_1$}
\psfrag{Z2}[t]{$s_2$}
\psfrag{Y1}[l][l]{$Y_1(t)$}
\psfrag{Y2}[l][l]{$Y_2(t)$}
\psfrag{p}[c]{$\mathrm{PP}(.)$}
\psfrag{a1}[b]{$A_1$}
\psfrag{a2}[b]{$A_2$}
\psfrag{t}[cc]{$\ t$}
\psfrag{x}[b]{$X(t)$}
\psfrag{y1}[b]{$\ Y_1(t)$}
\psfrag{y2}[b]{$\ Y_2(t)$}
\psfrag{0}[cc]{$0\ $}
\psfrag{1}[cc]{$1$}
\includegraphics[scale=0.5]{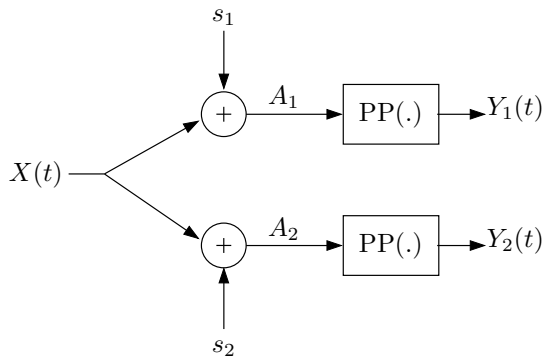}
\end{center}
\caption{Two receiver continuous time Poisson broadcast channel.}
\label{fig1}
\end{figure}

As in~\cite{Lapidoth--Telatar--Urbanke2003}, we use Wyner's approach for the point-to-point Poisson channel to study the capacity region of the Poisson broadcast channel. As depicted in Figure~\ref{fig:BPBC}-(a), time is quantized into intervals of length $\Delta>0$ and the continuous time P-BC is approximated in each interval by the binary memoryless broadcast channel (binary P-BC) depicted in Figure~\ref{fig:BPBC}-(b), with transition probabilities:
\begin{align}\begin{split}\label{BPBC}
a_1&=A_1s_1\Delta+O(\Delta^2), \quad a_2=A_2s_2\Delta+O(\Delta^2),\\ 
b_1&=A_1(1+s_1)\Delta+O(\Delta^2), \quad b_2=A_2(1+s_2)\Delta+O(\Delta^2).
\end{split}\end{align}
Following Wyner's arguments, Lapidoth, Telatar, and Urbanke showed that the capacity region of the P-BC is equal to the capacity region of the $1/\Delta$-extension of this binary channel as $\Delta$ tends to zero.

\begin{figure}[t]
\begin{center}
\begin{tabular}{ccccc}
\psfrag{0}[t]{$0$}
\psfrag{1}[t]{$1$}
\psfrag{x1}[c]{$1$}
\psfrag{x2}[c]{$2$}
\psfrag{x3}[c]{$3$}
\psfrag{x4}[c]{$4$}
\psfrag{xn}[c]{$1/\Delta$}
\psfrag{a}[l]{\small $t$}
\psfrag{a1}[t]{$\Delta$}
\psfrag{b}[c]{}
\includegraphics[scale=0.55]{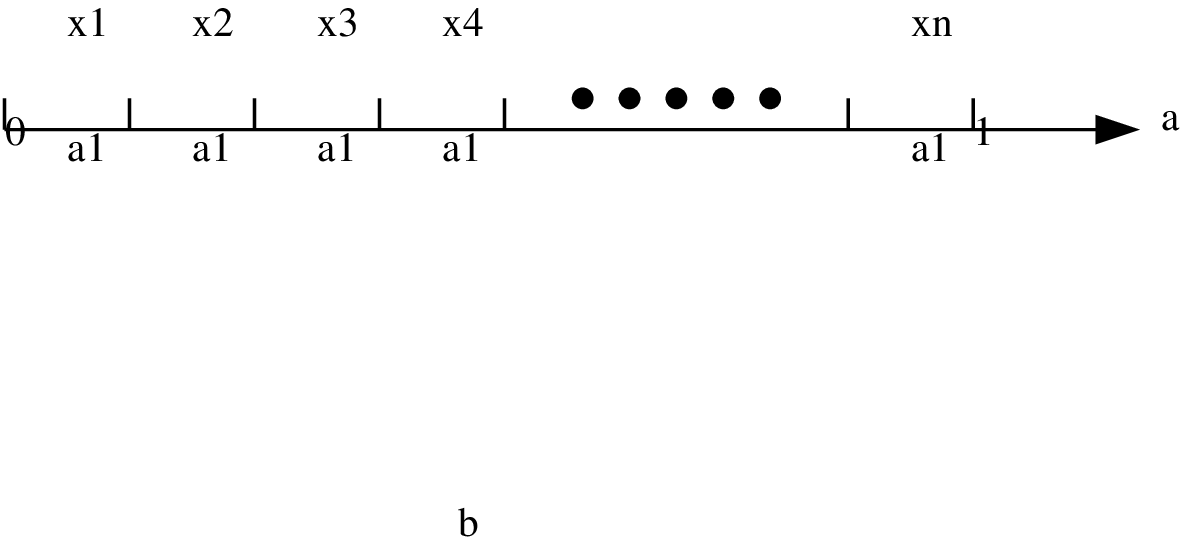}&&&&
\psfrag{0}[r]{$0$}
\psfrag{1}[r]{$1$}
\psfrag{z}[l]{$1$}
\psfrag{o}[l]{$0$}
\psfrag{z1}[l]{$0$}
\psfrag{o1}[l]{$1$}
\psfrag{q}[b]{$a_1$} 
\psfrag{p'}[t]{$b_1$} 
\psfrag{p}[b]{$a_2$} 
\psfrag{q'}[t]{$b_2$} 
\psfrag{X}[l]{$X$}
\psfrag{Y1}[c]{$Y_1^\Delta$}
\psfrag{Y2}[c]{$Y_2^\Delta$}
\includegraphics[scale=0.55]{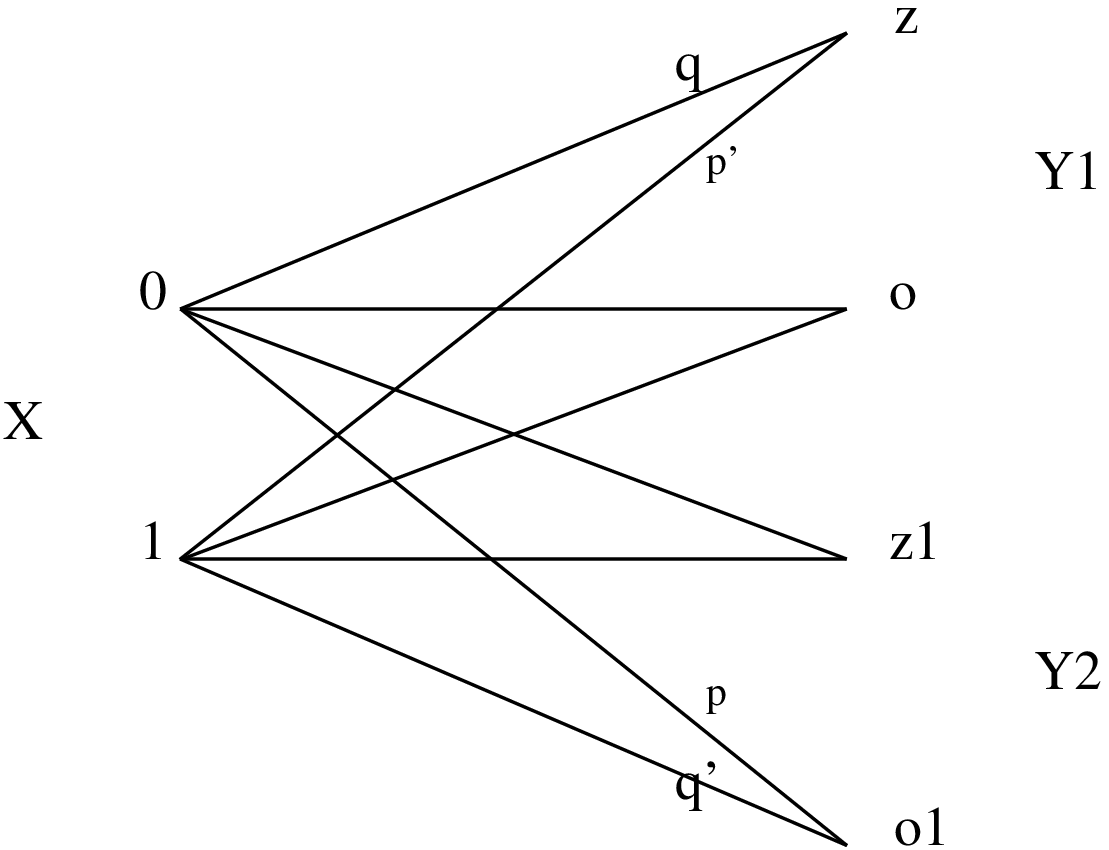}\\&&&&\\
(a) &&&& (b)
\end{tabular}
\caption{(a) Time quantization. (b) Binary P-BC.}
\label{fig:BPBC}
\end{center}
\end{figure} 


In this paper, we show the surprising fact that superposition coding is optimal for almost all channel parameter values. We find explicit analytical expressions for the ranges of parameter values in which the channel is less noisy and more capable~\cite{Korner--Marton1975}. We then introduce the new class of {\em effectively less noisy} broadcast channels for which superposition coding is optimal and show that it includes the less noisy class but is not strictly included in the more capable class. The key idea is that the less noisy condition needs to hold only for channel input distributions that attain the maximum weighted sum rate. We find explicit analytical expressions for the ranges of parameter values for which the P-BC is effectively less noisy. By further strengthening the effectively less noisy condition, we show numerically that superposition coding can be optimal even when the channel is not more capable or effectively less noisy. Finally we show that for the remaining set of parameter values. there is a gap between Marton's inner bound~\cite{Marton1979} and the UV outer bound~\cite{Nair--EG2007}.

The rest of the paper is organized as follows. In Section~\ref{sec:superposition} we review superposition coding and the optimal superposition coding inner bound for binary input broadcast channels, and summarize known classes of broadcast channels for which superposition coding inner bound is tight. In Section~\ref{sec:lnmc}, we establish the parameter ranges for which a P-BC is less noisy and more capable. 
In Section~\ref{sec:eff}, we introduce the new class of effectively less noisy broadcast channels for which superposition coding is also optimal, and obtain the parameter range for which a P-BC is effectively less noisy. 
In Section~\ref{sec:gap}, we show that for certain range of parameter values, there is a gap between Marton's inner bound and the UV outer bound. Hence, the capacity region of the P-BC is still not known in general. 
In Section~\ref{sec:extensions}, we extend our results to the average power constraint case. Finally in Section~\ref{sec:remarks}, we remark on the optimality of superposition coding for general binary input broadcast channels. We demonstrate via an example that our intuition about when superposition coding is optimal for the broadcast channel can be quite misleading. 



\section{Superposition coding inner bound}\label{sec:superposition}
Consider a 2-receiver discrete memoryless broadcast channel $p(y_1,y_2|x)$. The superposition coding scheme~\cite{Cover1972} is motivated by broadcast channels for which one receiver is ``stronger" than the other. This suggests a layered coding approach in which the weaker receiver (say $Y_2$) recovers only its own message $M_2$ carried by the auxiliary random variable $U$, while the stronger receiver ($Y_1$) recovers both messages $(M_1,M_2)$ carried by $X$. This coding scheme leads to the  inner bound on the capacity region of the general discrete memoryless broadcast (DM-BC) channel $p(y_1,y_2|x)$ that consists of all rate pairs $(R_1,R_2)$ such that~\cite{Bergmans1973}
\begin{align}\begin{split}\label{region:more capable}
R_1 &< I(X;Y_1|U),\\
R_2 &< I(U;Y_2),\\
R_1+R_2 &< I(X;Y_1)
\end{split}\end{align}
for some pmf $p(u,x)$, and $|\Ucal| \le |\Xcal|+1$. 

Let $\Rr_1$ denote the region~\eqref{region:more capable}. Note that a second superposition coding inner bound can be readily obtained by exchanging $Y_1$ and $Y_2$ and $R_1$ and $R_2$ in~\eqref{region:more capable}. Denoting the second region by $\Rr_2$,  it can be shown 
 that if the capacity of the channel $p(y_1|x)$, $C_1$, is larger than the capacity of the channel $p(y_2|x)$, $C_2$,  then $\Rr_2 \subseteq \{(R_1,R_2)\colon R_1/C_1+R_2/C_2 \le 1\}$.

\begin{remark}\textnormal{
In~\cite{Wang--Sasoglu--Bandemer--Kim2013}, Wang, Sasoglu, Bandemer, and Kim compared two superposition encoding schemes for the broadcast channel. The first is the $UX$ scheme above, and the second is the $(U,V)$ scheme in which $M_1$ is carried by $V$, $M_2$ is carried by $U$, and $X$ is a function of $(U,V)$. The weaker receiver $Y_2$ again recovers $M_2$ carried by $U$ and the stronger receiver $Y_1$ recovers $M_1$ carried by $V$. They showed that the optimal inner bound achieved by the $UV$ scheme can be strictly larger than  that achieved by the $UX$ scheme. It turns out, however, that if the broadcast channel has binary input, the optimal inner bound (assuming $C_1 \ge C_2$) is $\co\left(\Rr_1 \cup \{(0,C_2)\}\right)$~\cite{Nair--Kim--EG2015}. Since in this paper we are concerned with the binary P-BC, which has binary input and outputs, we focus only on conditions under which either $\Rr_1$ or $\Rr_2$ is optimal.}
\end{remark}

It is well known that region $\Rr_1$ is tight for the following classes of DM-BC. 

\begin{definition}[Degraded broadcast channel~\cite{Cover1972}]\textnormal{For a DM-BC $p(y_1, y_2|x)$ $Y_2$ is said to be a degraded version of $Y_1$  if there exists a random variable $Y_1'$ such that $Y_1'|\{X=x\} \sim p_{Y_1|X}(y'_1|x),$ i.e., $Y_1'$ has the same conditional pmf as $Y_1$ (given $X$) and $X \to Y_1' \to Y_2$ form a Markov chain.
}\end{definition}

\begin{definition}[Less noisy channel~\cite{Korner--Marton1975}]\textnormal{For a DM-BC $p(y_1,y_2|x)$ receiver $Y_1$ is said to be \emph{less noisy} than receiver $Y_2$ if $I(U;Y_1) \geq I(U;Y_2)$ for all $p(u,x)$. 
}\end{definition}
Van-Dijk~\cite{van-Dijk1997} showed that receiver $Y_1$ is less noisy than receiver $Y_2$ if $I(X;Y_1) - I(X;Y_2)$ is concave in $p(x)$, or equivalently, $I(X;Y_1)-I(X;Y_2)$ is equal to its upper concave envelope $\CC[I(X;Y_1) - I(X;Y_2)]$ (the smallest concave function that is greater than or equal to $I(X;Y_1)-I(X;Y_2)$). As we will see, this alternative condition of less noisy is significantly simpler to compute than the original condition.

\begin{definition}[More capable channel~\cite{Korner--Marton1975}]\textnormal{For a DM-BC $p(y_1,y_2|x)$ receiver $Y_1$ is said to be \emph{more capable} than receiver $Y_2$ if $I(X;Y_1) \geq I(X;Y_2)$ for all $p(x)$. 
}\end{definition}
The more capable condition can also be recast in terms of the concave envelope: Receiver $Y_1$ is more capable than $Y_2$ if $\CC[I(X;Y_2)-I(X;Y_1)]=0$ for every $p(x)$. 

It is also well known that degraded implies less noisy which implies more capable~\cite{Korner--Marton1975}, but the converses do not always hold. In~\cite{Nair2010}, Nair generalized the notions of less noisy and more capable. Let $\Pc_o$ be a class of pmfs $p(u,v,x)$ such that for any triple of random variables $(U,V,X) \sim p(u,v,x)$, there exists a pmf $q(\tilde{u},\tilde{v},x)$ such that 
\begin{align}\begin{split}\label{setPo}
q(x) &\in \Pc_o,\\
I(V;Y_1)_p &\le I(\tilde{V};Y_1)_q,\\
I(U;Y_2)_p &\le I(\tilde{U};Y_2)_q,\\
I(V;Y_1)_p +I(X;Y_2|V)_p &\le I(\tilde{V};Y_1)_q+I(X;Y_2|\tilde{V})_q ,\\
I(U;Y_2)_p+I(X;Y_1|U)_p &\le I(\tilde{U};Y_2)_q+I(X;Y_1|\tilde{U})_q.
\end{split}\end{align}
The notation $I(V;Y_1)_p$ refers to the mutual information between $V$ and $Y_1$ when the input is generated according to $p(u,v,x)$.

\begin{definition}[Essentially less noisy~\cite{Nair2010}]\textnormal{
For a DM-BC $p(y_1,y_2|x)$, $Y_1$ is said to be \emph{essentially less noisy} than $Y_2$  if there exists a sufficient class of pmfs $\Pc_o$ as defined in~\eqref{setPo} such that $I(U;Y_1) \geq I(U;Y_2)$ for every $p(x) \in \Pc_o$ and all $U \to X \to (Y_1,Y_2)$.
}\end{definition}

\begin{definition}[Essentially more capable~\cite{Nair2010}]\textnormal{
For a DM-BC $p(y_1,y_2|x)$, $Y_1$ is said to be \emph{essentially more capable} than $Y_2$  if there exists a sufficient class of pmfs $\Pc_o$ as defined in~\eqref{setPo} such that $I(X;Y_1|U) \geq I(X;Y_2|U)$ for every $p(x) \in \Pc_o$ and all $U \rightarrow X \rightarrow (Y_1,Y_2)$.
}\end{definition}

It can be easily seen from the definitions that less noisy implies essentially less noisy and more capable implies essentially more capable. However the converses do not always hold. Also it is shown in~\cite{Nair2010} that essentially less noisy neither implies nor is implied by essentially more capable. The capacity region for the essentially less noisy and the essentially more capable classes~\cite{EG1979} is region $\Rr_1$ as defined in~\eqref{region:more capable}. Also note that when the channel is essentially less noisy, the sum bound in~\eqref{region:more capable} is always inactive.

\section{Less noisy and more capable P-BC}\label{sec:lnmc}
We evaluate the conditions for the less noisy and more capable stated in the previous section for the P-BC and show that the P-BC is less noisy for almost all channel parameter values. We assume without loss of generality that $s_1 \le s_2$. The results for $s_1 > s_2$ can be similarly established. Also assume that $A_1 = \a$ and $A_2 = 1$. The result for $A_2 \neq 1$ can be established using the fact that the capacity for $(A_1,A_2) = (a_1,a_2)$ is equal to $a_2$ times the capacity region for $(A_1,A_2) = (a_1/a_2,1)$.

In~\cite{Lapidoth--Telatar--Urbanke2003}, Lapidoth et al. showed that $Y_2$ is a degraded version of $Y_1$ if $\a \ge 1$. This can be seen either by directly inspecting the P-BC channel model or its binary P-BC counterpart. 


We now determine the conditions under which the P-BC is less noisy and more capable.
 
We use the binary P-BC with the parameters given in~\eqref{BPBC} to extend the definitions of less noisy and more capable to the P-BC. Let
\[
I_i(q)= \lim_{\Delta \to 0} \frac{1}{\Delta} I(X;Y_i^{\Delta}),\; X \sim \Bern(q),\; i=1,2, 
\]
where $X$ and $Y_i^\Delta$ are the input and outputs of binary P-BC as depicted in Figure~\ref{fig:BPBC}-(b). 
Using Wyner's results~\cite{Wyner1988a, Wyner1988b}, it immediately follows that
\begin{align}\begin{split}\label{I_i}
I_1(q) &= \a \big(-(q+s_1)\log(q+s_1) + q(1+s_1)\log(1+s_1)+(1-q)s_1\log(s_1)\big),\\
I_2(q) &= -(q+s_2)\log(q+s_2) + q(1+s_2)\log(1+s_2)+(1-q)s_2\log(s_2),
\end{split}\end{align}
 and that $I_i(q)$ is maximized at 
 \begin{align}
 q_i= \frac{(1+s_i)^{1+s_i}}{es_i^{s_i}} -s_i\label{qi}
 \end{align}
 for $i=1,2$.

%


We define the less noisy condition for P-BC as follows. 
\begin{definition}\textnormal{For the 2-receiver P-BC, receiver $Y_1$ is less noisy than $Y_2$ if $I_1(q) - I_2(q)$ is concave in $q \in [0,1]$, i.e., if $\CC[I_1(q) - I_2(q)]=I_1(q) - I_2(q)$.}
\end{definition}
Similarly we define the more capable condition for P-BC as follows.
\begin{definition}\textnormal{For the 2-receiver P-BC, receiver $Y_1$ is more capable than $Y_2$ if $I_1(q) - I_2(q) \ge 0$ for every $q \in [0,1]$, i.e., if $\CC[I_2(q) - I_1(q)]=0$.
}\end{definition}


To establish the parameter ranges for less noisy and more capable P-BC, we define the following breakpoints of $\a$:
\begin{align*}
\a_1 &= \frac{1+s_1}{1+s_2},\\
\a_2 &= \frac{s_2 \log(1+1/s_2) -1}{s_1 \log(1+1/s_1)-1},\\
\a_3 &= \frac{(1+s_2)\log(1+1/s_2)-1}{(1+s_1)\log(1+1/s_1)-1},\\
\a_4 &= \frac{s_1}{s_2}.
\end{align*}

We also need the following lemma which characterizes the upper concave envelope of $(I_1(q) - I_2(q))$ and $(I_2(q) - I_1(q))$.
\begin{lemma}\label{lemma:envelope}\textnormal{Consider a P-BC.
\begin{enumerate} 
\item[1.]
Let
\begin{align}\label{def:t}\allowdisplaybreaks
t = \begin{cases}
0 &\text{ if } 0\le \a \leq \a_3,\\ 
g_1^{-1}(\a) &\text{ if } \a_3 < \a < \a_1,\\
1 &\text{ if } \a \geq \a_1,
\end{cases}
\end{align}
where
\begin{align}
g_1(x) = \frac{(1+s_2) \log \big((1+s_2)/(x+s_2)\big) - 1+ x}{(1+s_1) \log \big((1+s_1)/(x+s_1)\big) - 1+x}.\label{g1}
\end{align}
Then, the upper concave envelope of $(I_1(q)-I_2(q))$ is
\begin{align*} 
\mathfrak{C}[I_1(q)-I_2(q)]=\begin{cases}
I_1(q)-I_2(q) &\text{ for } 0\leq q \leq t,\\
(1-q)(I_1(t)-I_2(t))/(1-t)&\text{ for } q > t.
\end{cases}
\end{align*}
\smallskip
\item[2.]  Let
\begin{align}\label{def:r}
r = \begin{cases}
0 &\text{ if } 0\le \a \leq \a_4,\\ 
g_2^{-1}(\a) &\text{ if } \a_4 < \a < \a_2,\\
1 &\text{ if } \a \geq \a_2,
\end{cases}
\end{align}
where
\begin{align}
g_2(x) = \frac{s_2\log(1+x/s_2)-x}{s_1\log(1+x/s_1)-x}.\label{g2}
\end{align}
Then, the upper concave envelope of $(I_2(q)-I_1(q))$ is
\begin{align*} 
\mathfrak{C}[I_2(q)\!-\!I_1(q)]=\begin{cases}
(q/r)(I_2(r)-I_1(r))&\text{ for }0\leq q < r,\\
I_2(q)-I_1(q) &\text{ for }q \geq r.
\end{cases}
\end{align*}
\end{enumerate}
}\end{lemma}

The proof of Lemma~\ref{lemma:envelope} is in Appendix~\ref{proof:envelope}. Figure~\ref{fig:ap0}-(a) plots $t$ vs $\a$. The shaded area is where $\mathfrak{C}[I_1(q)-I_2(q)] = I_1(q)-I_2(q)$. Figure~\ref{fig:ap0}-(b) plots $r$ vs $\a$. The shaded region is where $\mathfrak{C}[I_2(q)-I_1(q)] = I_2(q)-I_1(q)$. Figure~\ref{illustration_lemma1} plots  $I_1(q)-I_2(q)$ and $\CC[I_1(q)-I_2(q)]$ for $\a = 0.3\a_3 + 0.7 g_1(q_2)$ 
and $\a = g_1(q_2)$ for $(s_1,s_2) = (0.1,1)$. Figure~\ref{illustration_lemma2} plots  $I_2(q)-I_1(q)$ and $\CC[I_2(q)-I_1(q)]$ for $\a = 0.3\a_4 + 0.7 g_2(q_1)$ 
 and $\a = g_2(q_1)$ for $(s_1,s_2) = (0.1,1)$.

\begin{figure}[htpb]
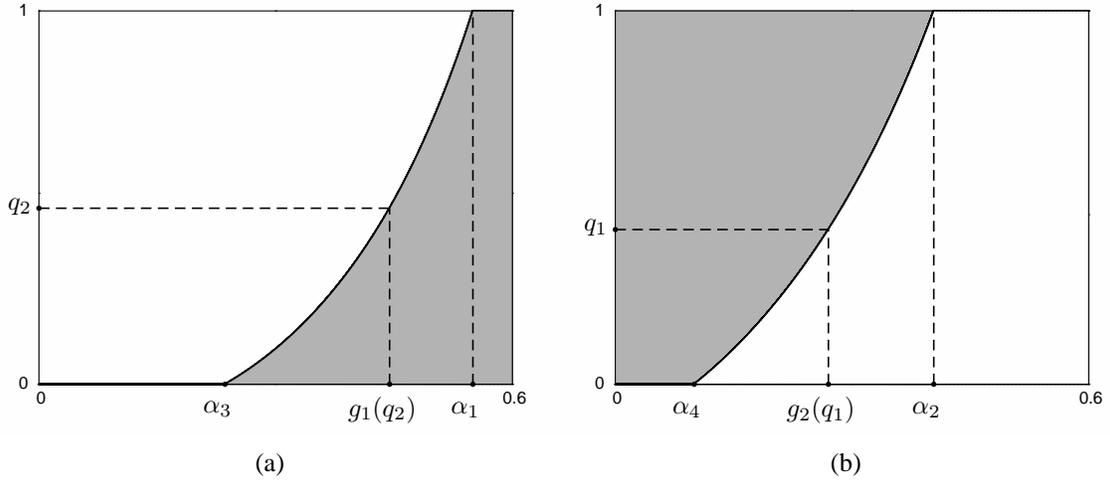

\begin{center}
\begin{tabular}{ccc}
\psfrag{y}[t]{$$}
\psfrag{0}[t]{$0$}\psfrag{1}[t]{$1$}\psfrag{0.6}[t]{$$}
\psfrag{q1}[t]{$~q_1$}
\psfrag{q2}[t]{$~q_2$}
\psfrag{a01}[t]{$\a_4$}
\psfrag{a1}[t]{$g_2(q_1)$}
\psfrag{a2}[t]{$\a_2$}
\psfrag{a02}[t]{$\a_3$}
\psfrag{a3}[t]{$g_1(q_2)$}
\psfrag{a31}[t]{$\a_1$}
\psfrag{M}[t]{$(a)$}
\psfrag{N}[t]{$(b)$}
\includegraphics[scale=0.4]{diff_v7_2.eps}&&
\psfrag{y}[t]{$$}
\psfrag{0}[t]{$0$}\psfrag{1}[t]{$1$}\psfrag{0.6}[t]{$$}
\psfrag{q1}[t]{$~q_1$}
\psfrag{q2}[t]{$~q_2$}
\psfrag{a01}[t]{$\a_4$}
\psfrag{a1}[t]{$g_2(q_1)$}
\psfrag{a2}[t]{$\a_2$}
\psfrag{a02}[t]{$\a_3$}
\psfrag{a3}[t]{$g_1(q_2)$}
\psfrag{a31}[t]{$\a_1$}
\psfrag{M}[t]{$(a)$}
\psfrag{N}[t]{$(b)$}
\includegraphics[scale=0.4]{diff_v7_1.eps}\\&&\\
(a) && (b)
\end{tabular}
\caption{(a) Plot of $t$ vs $\a$ for $s_1=0.1, s_2=1$. The shaded area is where $\mathfrak{C}[I_1(q)-I_2(q)]=I_1(q)-I_2(q)$, i.e., $0 \le q \le t$. (b) Plot of $r$ vs $\a$ for $s_1=0.1, s_2=1$. The shaded area is where $\mathfrak{C}[I_2(q)-I_1(q)]=I_2(q)-I_1(q)$, i.e., $r \le q \le 1$.}
\label{fig:ap0}
\end{center}
\end{figure}

\begin{figure}[htpb]
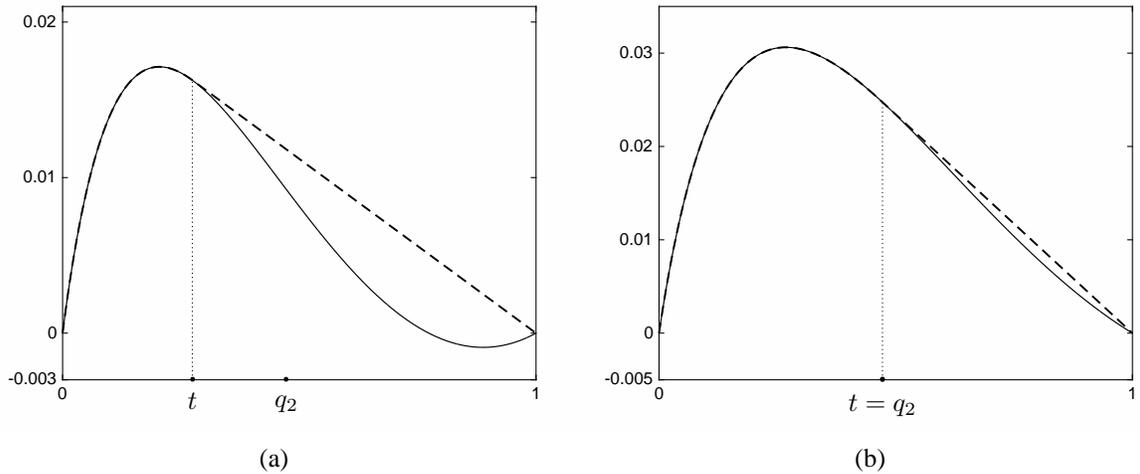

\begin{center}
\begin{tabular}{ccc}
\psfrag{a}[t]{$q$}
\psfrag{k}[t]{$q_2$}
\psfrag{q2}[t]{$t=q_2$}
\psfrag{t}[t]{$t$}
\includegraphics[scale=0.4]{illustration_lemma_2a_v4.eps}&&
\psfrag{a}[t]{$q$}
\psfrag{k}[t]{$q_2$}
\psfrag{q2}[t]{$t=q_2$}
\psfrag{t}[t]{$t$}
\includegraphics[scale=0.4]{illustration_lemma_2b_v4.eps}\\&&\\
(a) && (b)
\end{tabular}
\caption{Plots of $I_1(q)-I_2(q)$ (solid) and $\mathfrak{C}[I_1(q)-I_2(q)]$ (dashed line) vs $q$ for (a) $\a = 0.3\a_3+0.7 g_1(q_2)$ and (b) $\a = g_1(q_2)$ ($s_1=0.1, s_2=1$).}
\label{illustration_lemma1}
\end{center}
\end{figure}

\begin{figure}[htpb]
\begin{center}
\begin{tabular}{ccc}
\psfrag{a}[t]{$q$}
\psfrag{s}[t]{$r=q_1$}
\psfrag{q1}[t]{$q_1$}
\psfrag{r}[t]{$r$}
\includegraphics[scale=0.4]{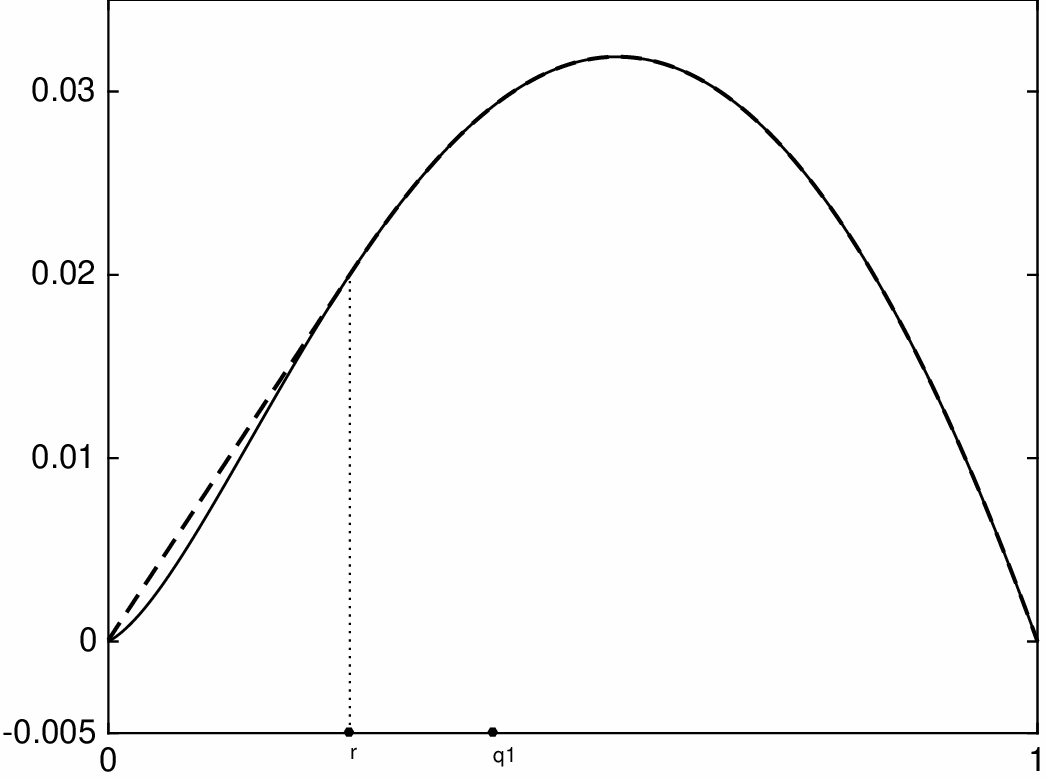}&&
\psfrag{a}[t]{$q$}
\psfrag{s}[t]{$r=q_1$}
\psfrag{q1}[t]{$q_1$}
\psfrag{r}[t]{$r$}
\includegraphics[scale=0.4]{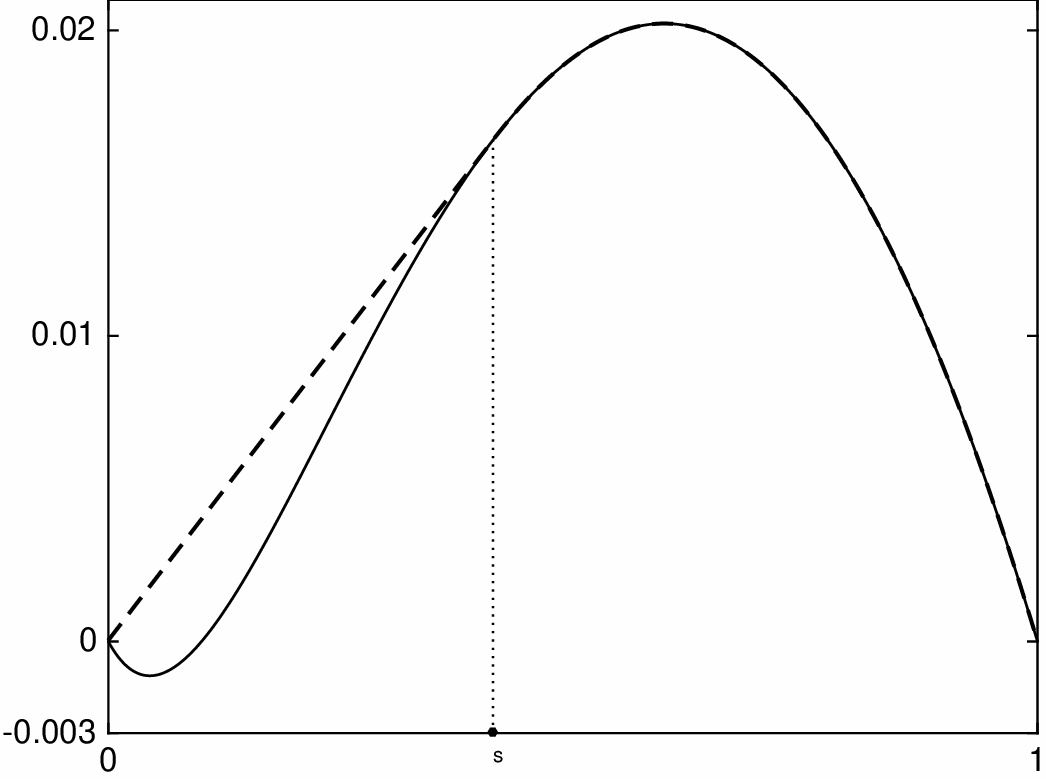}\\&&\\
(a) && (b) 
\end{tabular}
\caption{Plots of $I_2(q)-I_1(q)$ (solid) and $\mathfrak{C}[I_2(q)-I_1(q)]$ (dashed line) vs $q$ for (a) $\a=0.3\a_4+0.7 g_2(q_1)$ and (b) $\a=g_2(q_1)$ ($s_1=0.1, s_2=1$).}
\label{illustration_lemma2}
\end{center}
\end{figure}

We are  now ready to state the conditions for less noisy and more capable.
\begin{theorem}\label{thm:lnmc}
\textnormal{ For $s_1 \le s_2$:
\begin{enumerate}
\item[1.] If $\a \ge \a_1$, $Y_1$ is less noisy than $Y_2$ and the capacity region is the set of rate pairs $(R_1,R_2)$ such that
\begin{align}\begin{split}\label{thm2-part1}
R_1 &\le \beta I_1(p_0),\\
R_2 &\le I_2(\beta p_0+\bar{\beta}) - \beta I_2(p_0)
\end{split}\end{align}
for some $0 \le \beta, p_0 \le 1$. 
\item[2.] If $\a_2 \le \a \le \a_1$, $Y_1$ is more capable than $Y_2$ and the capacity region is the set of rate pairs $(R_1,R_2)$ such that
\begin{align}\begin{split}\label{thm2-part2}
R_1 &\leq \beta_0 I_1(p_0) + \beta_1 I_1(p_1) + \beta_2 I_1(p_2),\\
R_2 &\leq I_2(p) - \beta_0 I_2(p_0) - \beta_1 I_2(p_1) - \beta_2 I_2(p_2),\\
R_1+R_2 &\leq I_1(p)
\end{split}\end{align}
for some $0 \le \beta_0, \beta_1, \beta_2, p_0, p_1, p_2 \le 1$, where $\beta_0+\beta_1+\beta_2=1$ and $p=\beta_0 p_0 + \beta_1 p_1 + \beta_2 p_2$.\\
\item[3.] If $0 \le \alpha \le \a_4$, $Y_2$ is less noisy than $Y_1$ and the capacity region is the set of rate pairs $(R_1,R_2)$ such that
\begin{align}\begin{split}\label{thm2-part3}
R_1 &\le I_1(\beta p_0) - \beta I_1(p_0),\\
R_2 &\le \beta I_2(p_0)
\end{split}\end{align}
for some $0 \le \beta, p_0 \le 1$. 
\item[4.] If $\a_4 \le \alpha \le \a_3$, $Y_2$ is more capable than $Y_1$ and the capacity region is the set of rate pairs $(R_1,R_2)$ such that
\begin{align*}
R_1 &\leq I_1(p) - \beta_0 I_1(p_0) - \beta_1 I_1(p_1) - \beta_2 I_1(p_2),\\
R_2 &\leq \beta_0 I_2(p_0) + \beta_1 I_2(p_1) + \beta_2 I_1(p_2),\\
R_1+R_2 &\leq I_2(p)
\end{align*}
for some $0 \le \beta_0, \beta_1, \beta_2, p_0, p_1, p_2 \le 1$, where $\beta_0+\beta_1+\beta_2=1$ and $p=\beta_0 p_0 + \beta_1 p_1 + \beta_2 p_2$.
\end{enumerate}
}\end{theorem}

The capacity region for $Y_1$ less noisy than $Y_2$ is the superposition rate region $\Rr_1$ without the sum rate bound for $U \sim \Bern(\beta)$ and $p(x|u)$ is the Z-channel shown in Figure~\ref{z-ch1}-(a) where $\beta, p_0 \in [0,1]$; and the capacity region for $Y_2$ less noisy than $Y_1$ is $\Rr_2$ for $U \sim \Bern(\beta)$ and $p(x|u)$ is the Z-channel shown in Figure~\ref{z-ch1}-(b) where $\beta, p_0 \in [0,1]$. 

\begin{figure}[htpb]
\begin{center}
\begin{tabular}{cccc}
\psfrag{0}[r]{$1$}
\psfrag{1}[r]{$0$}
\psfrag{00}[l]{$1$}
\psfrag{11}[l]{$0$}
\psfrag{p}[b]{$\bar{p_0}$}
\psfrag{p'}[t]{$$}
\psfrag{x}[c]{$U$}
\psfrag{y}[c]{$X$}
\includegraphics[scale=0.5]{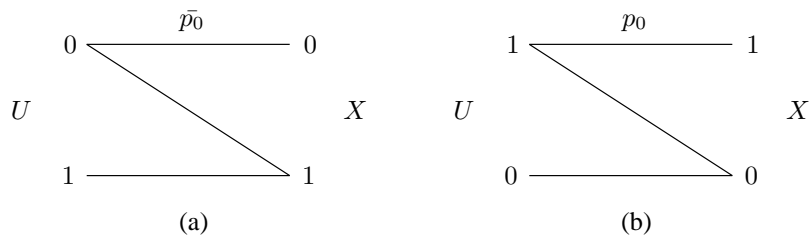}&&&
\psfrag{0}[r]{$0$}
\psfrag{1}[r]{$1$}
\psfrag{00}[l]{$0$}
\psfrag{11}[l]{$1$}
\psfrag{p}[b]{$p_0$}
\psfrag{p'}[t]{$$}
\psfrag{x}[c]{$U$}
\psfrag{y}[c]{$X$}
\includegraphics[scale=0.5]{z-channel.eps}\\
(a) &&& (b)
\end{tabular}
\end{center}
\caption{(a) $p(x|u)$ for $Y_1$ less noisy than $Y_2$. (b)  $p(x|u)$ for $Y_2$ less noisy than $Y_1$.}\label{z-ch1}
\end{figure}

The capacity region for $Y_1$ more capable than $Y_2$ is the superposition rate region $\Rr_1$ with $U \in \{0,1,2\}$, $p_U(j)=\beta_j$, and $X|\{U=j\} \sim \Bern(p_j)$, $j=0,1,2$; and the capacity region for $Y_2$ more capable than $Y_1$ is the superposition rate region i$\Rr_2$ with $p_U(j)=\beta_j$ and $X|\{U=j\} \sim \Bern(p_j)$ for $j=0,1,2$.

Figure~\ref{fig:thm2thm3} illustrates the ranges of $\a$ for which the P-BC is degraded, less noisy and more capable. From Theorem~\ref{thm:lnmc} the P-BC is less noisy if $\a \le \a_4$ or $\a \ge \a_1$ and more capable if $\a \le \a_3$ or $\a \ge \a_2$. 

\begin{figure}[htpb]
\begin{center}
\psfrag{0}[t]{}
\psfrag{1}[b]{}
\psfrag{2}[b]{}
\psfrag{3}[b]{}
\psfrag{4}[b]{}
\psfrag{5}[b]{}
\psfrag{6}[b]{}
\psfrag{7}[b]{}
\psfrag{8}[b]{}
\psfrag{a}[t]{$\a$}
\psfrag{a1}[t]{$\a_4$}
\psfrag{a2}[t]{$\a_3$}
\psfrag{a3}[t]{}
\psfrag{a4}[t]{$\a_2$}
\psfrag{a5}[t]{}
\psfrag{a6}[t]{$\a_1$}
\psfrag{a7}[t]{$1$}
\psfrag{m}[r]{\footnotesize more capable\hspace{-2pt}}
\psfrag{l}[r]{\footnotesize less noisy}
\psfrag{l1}[l]{\footnotesize less noisy}
\psfrag{d}[r]{\footnotesize degraded}
\includegraphics[scale=0.65]{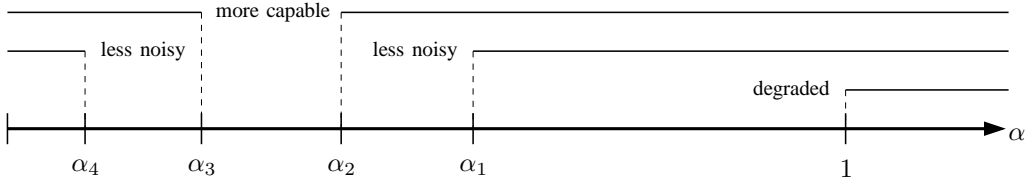}
\end{center}
\caption{Illustration of the ranges of $\a$ for which the channel is degraded, less noisy, and more capable as stated in Theorem~\ref{thm:lnmc}.}
\label{fig:thm2thm3}
\end{figure}

\begin{proof}[Proof of Theorem~\ref{thm:lnmc}]

We first find the condition on $\a$ for one receiver being less noisy or more capable than the other receiver.
\begin{enumerate}
\item[1.] By part 1 of Lemma~\ref{lemma:envelope}, $\CC[I_1(q)-I_2(q)] = I_1(q)-I_2(q)$ if and only if $\a \ge \a_1$. Thus $Y_1$ is less noisy than $Y_2$ iff $\a \ge \a_1$. 
\item[2.] By part 2 of Lemma~\ref{lemma:envelope}, $\CC[I_2(q)-I_1(q)] = 0$ if and only if $\a \ge \a_2$. Thus $Y_1$ is more capable than $Y_2$ iff  $\a \ge \a_2$. 
\item[3.] By part 2 of Lemma~\ref{lemma:envelope}, $\CC[I_2(q)-I_1(q)] = I_2(q)-I_1(q)$ if and only if $0 \le \a \le \a_4$. Thus $Y_2$ is less noisy than $Y_1$ iff $0 \le \a \le \a_4$.
\item[4.] By part 1 of Lemma~\ref{lemma:envelope}, $\CC[I_1(q)-I_2(q)] = 0$ if and only if $0 \le \a \le \a_3$. Thus $Y_2$ is more capable than $Y_1$ iff $0 \le \a \le \a_3$. 
\end{enumerate}
We now obtain the capacity region for P-BC such that $Y_1$ is more capable than $Y_2$ (part 2) and for P-BC such that $Y_1$ is less noisy than $Y_2$ (part 1). By exchanging $Y_1$ and $Y_2$ and $I_1(\cdot)$ and $I_2(\cdot)$, the capacity expression for part 4  and part 3 of Theorem~\ref{thm:lnmc} can be obtained similarly. 

For the $1/\Delta$-extension binary P-BC shown in Figure~\ref{fig:BPBC}-(b), superposition coding inner bound is the set of rate pairs $(R_1,R_2)$ such that
\begin{align}\begin{split}\label{region:more capable-PBC}
R_1 &< (1/\Delta)I(X;Y^\Delta_1|U),\\
R_2 &< (1/\Delta)I(U;Y^\Delta_2),\\
R_1+R_2 &< (1/\Delta)I(X;Y_1^\Delta)
\end{split}
\end{align}
for some pmf $p(u,x)$, and $|\Ucal| \le |\Xcal|+1$.

Let $U \in \{0,1,2\}$ where $p_U(i) = \beta_i$ and $p_{X|U}(1|i) = p_i \in [0,1]$. As $\Delta \to 0$, the region in~\eqref{region:more capable-PBC} is equivalent to the region in~\eqref{thm2-part2} in Theorem 2. If a P-BC is more capable, the region in~\eqref{thm2-part2} is indeed the capacity region. This is because the UV outer bound for $1/\Delta$-extension binary P-BC as $\Delta \to 0$ reduces to the inner bound in~\eqref{thm2-part2}  under the condition for more capable P-BC.

For part 1 of Theorem~\ref{thm:lnmc}, note that the capacity region is the region in~\eqref{thm2-part2} with an inactive sum bound, i.e., the set of rate pairs $(R_1,R_2)$ that satisfy 
\begin{align}\begin{split}\label{region1}
R_1 &\le \sum_{i\in\{0,1,2\}} \beta_i I_1(p_i),\\
R_2 &\le I_2\bigg(\sum_{i\in\{0,1,2\}}\beta_i p_i\bigg) - \sum_{i\in\{0,1,2\}}\beta_i I_2(p_i)
\end{split}\end{align}
for some $\beta_i, p_i \in [0,1]$ such that $\sum_{i\in\{0,1,2\}}\beta_i = 1$.
Let $\Rr'$ and $\Rr''$ denote the region in~\eqref{thm2-part1} and ~\eqref{region1} respectively. 
Note that $\Rr'$ is the set of rate pairs $(R_1,R_2)$ that satisfy inequalities in~\eqref{region1} for $(\b_0, \b_1, \b_2) = (\beta, \bar{\beta}, 0)$ for some $\beta \in [0,1]$ and $p_0 \in [0,1]$ and $p_1 = 1$. Thus $\Rr' \subseteq \Rr''$. We now show that every supporting hyperplane of $\Rr''$ intersects $\Rr'$, i.e., $\max_{(R_1,R_2) \in \Rr''}(\l R_1+R_2) \le \max_{(r_1,r_2) \in \Rr'} (\l r_1+r_2)$. Consider 
\begin{align}
\max_{(R_1,R_2) \in \Rr''} (\lambda R_1 + R_2) &=\max_{\beta_i, p_i \in [0,1]} \bigg(\sum_{i\in\{0,1,2\}} \beta_i (\l I_1(p_i) - I_2(p_i)) + I_2\big(\sum_{i\in\{0,1,2\}}\beta_i p_i\big)\bigg)\nonumber\\
					 &=\max_{ p \in [0,1]} \bigg(\max_{\b_i,p_i: \sum \b_i p_i = p} \Big( \sum_{i\in\{0,1,2\}} \beta_i (\l I_1(p_i) - I_2(p_i))\Big) + I_2(p)\bigg)\nonumber\\
					 &=\max_{ p \in [0,1]} \bigg( \CC[\l I_1(p) - I_2(p)] + I_2(p)\bigg)\nonumber\\
					 &\stackrel{(a)}{=}\max \bigg(\max_{p\in [0,t']} \big(\l I_1(p)-I_2(p) + I_2(p) \big), \max_{p\in[t',1]} \big((1-p)(\l I_1(t')-I_2(t'))/(1-t') + I_2(p)\big) \bigg)\label{above}\\
					 &\stackrel{(b)}{\le} \max_{\beta,p_0 \in [0,1]} \big(\b(\l I_1(p_0) - I_2(p_0)) + I_2(\b p_0+\bar{\b})\big)\label{below}\\ 
					 &=\max_{(r_1,r_2) \in \Rr'} (\lambda r_1 + r_2).\nonumber
\end{align}
Step $(a)$ holds by Lemma~\ref{lemma:envelope}. To prove step $(b)$, note that the two terms in~\eqref{above} are obtained by letting $(\b, p_0)=(1, p)$ and $(\b, p_0)=((1-p)/(1-t'), t')$ in~\eqref{below}.

We now show that $\Rr'$ is convex. Suppose $\Rr'$ is not convex, i.e., $\co\{\Rr'\} \neq \Rr'$ where $\co\{\Rr'\}$ denotes the convex hull of $\Rr'$.  There exists a rate pair $(r_1,r_2) \in \co\{\Rr'\}$ on the boundary of $\co\{\Rr'\}$, i.e., $\mu r_1+r_2 = \max_{(R_1,R_2) \in \co\{\Rr'\}}(\mu R_1+R_2)$ for some $\mu \ge 0$ such that $(r_1,r_2) \not\in \Rr'$.
Note that $(r_1,r_2) = \eta(r_{10},r_{20})+(1-\eta)(r_{11},r_{21})$ for some $0<\eta<1$ and $(r_{10},r_{20}), (r_{11},r_{21}) \in \Rr'$. Since $\max_{(R_1,R_2) \in \co\{\Rr'\}}(\mu R_1+R_2) = \max_{(R_1,R_2) \in \Rr'}(\mu R_1+R_2)$, the two rate pairs $(r_{10},r_{20})$ and $(r_{11},r_{21})$ satisfy
\begin{align}
\mu r_{10}+r_{20}= \mu r_{11}+r_{21} = \max_{(R_1,R_2) \in \Rr'} (\mu R_1+R_2).\label{two-pairs}
\end{align}
We now show that the equality~\eqref{two-pairs} cannot hold for $(r_{10},r_{20}) \neq (r_{11},r_{21})$, i.e., there exists a unique rate pair $(R_1,R_2)$ such that $\mu R_1+R_2 = \max_{(r_1,r_2) \in \Rr'} (\mu r_1+r_2)$. 
Consider 
\begin{align*}
\max_{(r_1,r_2) \in \Rr'} \big(\mu r_1+r_2\big)=\max_{\b, p_0 \in [0,1]} \big(I_2(\b p_0+\bar{\b}) + \mu \b I_1(p_0) - \b I_2(p_0)\big).
\end{align*}
We show that $(\b, p_0)$ that achieves the maximum is unique. Note that 
\[
\max_{\b, p_0} \big(I_2(\b p_0+\bar{\b}) + \mu \b I_1(p_0) - \b I_2(p_0)\big) = \max_p \big(I_2(p) + \CC[\mu I_1(p) - I_2(p)]\big).
\]
Since $I_2(p)$ is strictly concave, $I_2(p) + \CC[\mu I_1(p) - I_2(p)]$ is strictly concave. Let $p^*$ the unique solution that maximizes $I_2(p) + \CC[\mu I_1(p) - I_2(p)]$. Then 
\begin{align*}
\max_{\b, p_0} \big(I_2(\b p_0+\bar{\b}) + \mu \b I_1(p_0) - \b I_2(p_0)\big) &= I_2(p^*) + \CC[\mu I_1(p^*) - I_2(p^*)]\\
&= \max_{\b, p_0: \b p_0+\bar{\b} = p^*} \big(I_2(p^*) + \mu \b I_1(p_0) - \b I_2(p_0)\big) 
\end{align*}
Finally by Lemma~\ref{lemma:envelope}, there exists a unique $(\beta,p_0)$ such that $\beta p_0+\bar{\beta} = p^*$ and 
\[
 \CC[\mu I_1(p^*) - I_2(p^*)] =  \mu \b I_1(p_0) - \b I_2(p_0).
\]
Thus $\Rr'=\co\{\Rr'\}$. To complete the proof for $\Rr' = \Rr''$, we use Lemma~\ref{subseteq} below.
\end{proof}

\smallskip
\begin{lemma}\label{subseteq}
\textnormal{~\cite{Eggleston1958} Let $\Rr \in \mathbb{R}^d$ be convex and $\Rr' \subseteq \Rr''$ be two bounded convex subsets of $\Rr$, closed relative to $\Rr$. If every supporting hyperplane of $\Rr''$ intersects $\Rr'$, then $\Rr' = \Rr''$. 
}\end{lemma}

As mentioned in the first part of this section, superposition coding is also optimal for the essentially less noisy and essentially more capable classes. Can we extend the range of parameter $\a$ for which superposition coding is optimal by evaluating the conditions for these two classes? 

To answer this question, first note that or binary input broadcast channels, the essentially more capable condition in~\cite{Nair2010} reduces to that for the more capable class; hence essentially more capable does not extend the range of $\a$ for which superposition coding is optimal beyond more capable. 
To see this, consider a binary input broadcast channel which is not more capable. Then there exists $p,q\in(0,1)$ such that $I(X;Y_2)_{p}-I(X;Y_1)_{p} > 0$ and $I(X;Y_1)_{q} - I(X;Y_2)_{q} > 0$. Then
\begin{align*}
\mathfrak{C}[I(X;Y_2)_r-I(X;Y_1)_r] > 0 &\text{ for every } r \in (0,1),\\
\mathfrak{C}[I(X;Y_1)_r-I(X;Y_2)_r] > 0 &\text{ for every } r \in (0,1).
\end{align*}
Thus, for every $X \sim \Bern(r)$ for $r \in (0,1)$, there exists $U_1$ and $U_2$ such that $I(X;Y_2|U_1)-I(X;Y_1|U_1)>0$ and $I(X;Y_1|U_2)-I(X;Y_2|U_2)>0$. Hence if a binary input broadcast channel is not more capable it is also not essentially more capable. 

The answer for essentially less noisy is less clear. It appears to be quite difficult to evaluate the set of pmfs $\Pc_o$ that satisfy the condition in~\eqref{setP}. In the following section, we define a new class of broadcast channels for which the condition can be easily evaluated and which includes the essentially less noisy class. 


\section{Effectively less noisy P-BC}\label{sec:eff}
Consider outer bound $\bar\Rr_1$ on the capacity region of the DM-BC $p(y_1,y_2|x)$ which consists of all rate pairs $(R_1,R_2)$ such that
 \begin{align}\begin{split}\label{outer}
R_2 &\leq I(U;Y_2),\\
R_1+R_2 &\leq I(U;Y_2)+I(X;Y_1|U)
\end{split}\end{align}
for some $p(u,x)$. To see that this is indeed an outer bound, note that it is simply the K{\"o}rner--Marton~\cite{Marton1979} outer bound without the sum rate bound. 

The outer bound $\bar \Rr_1$ can be alternatively represented in terms of its supporting hyperplanes as: 
\begin{align}\begin{split}\label{hyper-outer}
\max_{(r_1,r_2) \in \bar \Rr_1}(\lambda r_1+ r_2) =\begin{cases}
\max_{p(u,x)}(\lambda I(X;Y_1|U)+I(U;Y_2))&\text{ if }0 \le \lambda \leq 1,\\
 \max_{p(u,x)}\lambda(I(U;Y_2)+I(X;Y_1|U)) &\text{ if }\lambda > 1.
\end{cases}
\end{split}\end{align}
Now consider the supporting hyperplane representation of the superposition rate region $\Rr_1$ in~\eqref{region:more capable}:
\begin{align*}
\max_{(R_1,R_2) \in \Rr_1} (\lambda R_1 + R_ 2) = 
\begin{cases} \max_{p(u,x)} \big(\lambda \min\{ I(X;Y_1|U), I(X;Y_1) -I(U;Y_2)\} + I(U;Y_2)\big) &\text{ if }0 \le \lambda \leq 1,\\
\max_{p(u,x)}\lambda \big(I(U;Y_1)+I(X;Y_1|U)\big) &\text{ if }\lambda > 1.
\end{cases}
\end{align*}
Note that $\bar \Rr_1$ and $\Rr_1$ differ only in the first term when $0 \le \lambda \le 1$: for $\bar \Rr_1$, the term is $\l I(X;Y_1|U)$, while for $\Rr_1$ the term is $\l \min\{I(X;Y_1|U),I(X;Y_1)-I(U;Y_2)\}$. When $\lambda > 1$, for $\bar \Rr_1$, the term is $\l I(U;Y_2)$, while for $\Rr_1$ the term is $\l I(U;Y_1)$.

Now it is easy to see that if the DM-BC is less noisy, i.e., $I(U;Y_1) \ge I(U;Y_2)$ for all $p(u,x)$, then these two bounds coincide. The key observation that leads to a more general class than less noisy is that the inequality $I(U;Y_1) \ge I(U;Y_2)$ does not need to hold for every $p(u,x)$. 
For example, if $I(U;Y_1) \ge I(U;Y_2)$ for every $p(x) \in \Pc$ and every $p(u|x)$ such that
\begin{align}\begin{split}\label{setPP}
\max_{p(x)}\max_{p(u|x)}(\lambda I(X;Y_1|U)+I(U;Y_2)) = \max_{p(x) \in \Pc}\max_{p(u|x)} (\lambda I(X;Y_1|U)+I(U;Y_2))
\end{split}\end{align}
for every $0 \le \l \le 1$, then $\Rr_1$ and $\bar \Rr_1$ coincide. 

This particular example is quite interesting because both the condition for the set $\Pc$ and the inequality $I(U;Y_1) \ge I(U;Y_2)$ can be expressed in terms of the upper concave envelope of $I(X;Y_1) - I(X; Y_2)$, which makes their evaluation quite straightforward especially for binary broadcast channels.

We are now ready to introduce a new class of broadcast channels for which superposition coding is optimal.
\begin{definition}[Effectively less noisy broadcast channels]\label{def:effectively}
\textnormal{
For a DM-BC, $p(y_1,y_2|x)$, let $\Pc$ be the set of pmfs $p(x)$ such that for every $0 \le \lambda \le 1$,}
\begin{align}\label{setP}
\begin{split}
\max_{p(x)} \big( I(X;Y_2) + \CC[\l I(X;Y_1) - I(X;Y_2)]\big) = \max_{p(x) \in \Pc}\big( I(X;Y_2) + \CC[\l I(X;Y_1) - I(X;Y_2)]\big). 
\end{split}
\end{align}
\textnormal{Receiver $Y_1$ is said to be {\em effectively less noisy} than receiver $Y_2$ if $I(X;Y_1) - I(X;Y_2)=\CC[I(X;Y_1) - I(X;Y_2)]$ for every $p(x)\in \Pr$. 
}\end{definition}

%

Clearly if the DM-BC is less noisy, then it is effectively less noisy. We can further show that if the channel is essentially less noisy as defined in~\cite{Nair2010}, it is also effectively less noisy. To show this note that  the sufficient class $\Pc_o$ in~\eqref{setPo} must satisfy
\begin{align*}
\max_{p(x)}\max_{p(u|x)}(\lambda I(X;Y_1|U)+I(U;Y_2)) \le \max_{p(x) \in \Pc_o}\max_{p(u|x)} (\lambda I(X;Y_1|U)+I(U;Y_2))
\end{align*}
for every $0 \le \l \le 1$. Hence $\Pc_o \supseteq \Pc$. If $I(U;Y_1) \leq I(U;Y_2)$ for $p(x)\in \Pc_o$ and every $p(u|x)$, then $I(U;Y_1) \leq I(U;Y_2)$ for $p(x)\in \Pc$ and every $p(u|x)$. We do not know if the condition for effectively less noisy is strictly weaker than that for essentially less noisy, however. As we will see in the next section, effectively less noisy neither  implies nor is implied by more capable in general. 


The definition of effectively less noisy can be readily extended to the P-BC in the same manner as the less noisy and more capable we presented in the previous section. 
\begin{definition}[Effectively less noisy P-BC]\label{def:effectivelyP-BC}
\textnormal{For the 2-receiver P-BC, let $\Qcal \subseteq [0,1]$ be such that for every $0\le \l \le 1$, 
\begin{align}
\max_{q\in[0,1]} I_2(q) + \CC[\l I_1(q) - I_2(q)] = \max_{q\in\Qcal} I_2(q) + \CC[\l I_1(q) - I_2(q)].\label{eff-pbc}
\end{align}
Receiver $Y_1$ is said to be effectively less noisy than $Y_2$ if $I_1(q) - I_2(q) = \CC[I_1(q) - I_2(q)]$ for every $q \in \Qcal$.
}\end{definition}

To establish the parameter ranges for effectively less noisy P-BC, we need the following additional breakpoints of $\a \colon$
\begin{align*}
\a_{12} &=g_1(q_2),\\ 
\a_{23} &=g_2(q_1) 
\end{align*}
where the functions $g_1(\cdot)$ and $g_2(\cdot)$ are defined in~\eqref{g1} and~\eqref{g2}, respectively.
  
\begin{theorem}\label{thm:eff}
\textnormal{Consider a 2-receiver P-BC and assume that $s_1 \le s_2$.
\begin{enumerate}
\item[1.] If $\a \ge \a_{12}$, $Y_1$ is effectively less noisy than $Y_2$, and the capacity region is the set of rate pairs that satisfy~\eqref{thm2-part1}.
\item[2.] If $\a \le \a_{23}$, $Y_2$ is effectively less noisy than $Y_1$, and the capacity region is the set of rate pairs that satisfy~\eqref{thm2-part3}.
\end{enumerate}
}\end{theorem}

Note that as for the less noisy case, the capacity region for effectively less noisy channels is also attained using binary $U$ and a Z-channel from $U$ to $X$.
Figure~\ref{fig:line} illustrates the parameter ranges for which superposition coding is optimal. 
\begin{figure}[htpb]
\begin{center}
\psfrag{0}[t]{}
\psfrag{1}[b]{}
\psfrag{2}[b]{}
\psfrag{3}[b]{}
\psfrag{4}[b]{}
\psfrag{5}[b]{}
\psfrag{6}[b]{}
\psfrag{7}[b]{}
\psfrag{8}[b]{}
\psfrag{a}[t]{$\a$}
\psfrag{a1}[t]{$\a_4$}
\psfrag{a2}[t]{$\a_3\ $}
\psfrag{a3}[t]{$\ \a_{23}$}
\psfrag{a4}[t]{$\a_2$}
\psfrag{a5}[t]{$\a_{12}$}
\psfrag{a6}[t]{$\a_1$}
\psfrag{a7}[t]{$1$}
\psfrag{e}[r]{\footnotesize eff. less noisy}
\psfrag{m}[r]{\footnotesize more capable\hspace{-2pt}}
\psfrag{l}[r]{\footnotesize less noisy}
\psfrag{l1}[l]{\footnotesize less noisy}
\psfrag{d}[r]{\footnotesize degraded}
\includegraphics[scale=0.65]{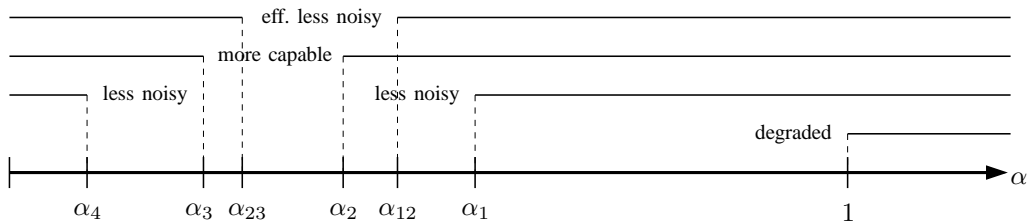}
\end{center}
\caption{Illustration of the ranges of $\a$ for which the channel is degraded, less noisy, more capable, and effectively less noisy as stated in Theorems~\ref{thm:lnmc} and~\ref{thm:eff}.}
\label{fig:line}
\end{figure}

It can be shown that $\a_4 \leq \a_3 \leq \a_{23} \leq \a_{2} \leq \a_{12} \leq \a_1$ (see Appendix~\ref{order}). The fact that $\a_2 \le \a_{12}$ implies that if $Y_1$ is effectively less noisy than $Y_2$, it is also more capable than $Y_2$ (see Figure~\ref{illustration_lemma1}-(b) for an example). Thus effectively less noisy does not offer any new range of parameters for which superposition coding is optimal. On the other hand, the fact that $\a_3 \le \a_{23}$ implies that for $\a \in (\a_3,\a_{23}]$, $Y_2$ is effectively less noisy but is not more capable than $Y_1$ (see Figure~\ref{illustration_lemma2}-(b) for an example). Thus effectively less noisy offers a new range of parameters for which superposition coding is optimal.
\smallskip

\begin{proof}[Proof of Theorem~\ref{thm:eff}] 
\begin{enumerate}
\item[1.]
We first show that $\Qc = [0,q_2]$ satisfies~\eqref{eff-pbc}. Specifically we show that 
$I_2(q)+ \mathfrak{C}[ \lambda I_1(q)- I_2(q)]$ is decreasing in $[q_2,1]$. Since $I_2(q)+ \mathfrak{C}[ \lambda I_1(q)- I_2(q)]$ is concave, it suffices to show that the derivative of $I_2(q)+ \mathfrak{C}[ \lambda I_1(q)- I_2(q)]$ at $q=q_2$ is nonpositive. Consider

\begin{align*}
\frac{d \big(I_2(q)+ \CC[\l I_1(q) - I_2(q)]\big)}{dq} \bigg |_{q=q_2} &\stackrel{(a)}{=}\begin{cases}
\l I_1'(q_2) &\text{ if } q_2 \ge t'\\
I_2'(q_2)-(\l I_1(t')-I_2(t'))/(1-t') &\text{ otherwise.}
\end{cases} \\
&\stackrel{(b)}{\le} 0
\end{align*}
for $t'=t(\a')$ where $\a ' = \l \alpha$. Step $(a)$ follows by rewritting $\l I_1(q)$ for $A_1 = \a$ as $I_1(q)$ for $A_1 = \l \a$ (see~\eqref{I_i}) and then applying Lemma~\ref{lemma:envelope}. Step $(b)$ holds because $\l I_1'(q_2) \le 0$ and $-(\l I_1(t')-I_2(t'))/(1-t')\le 0$. The first inequality holds since $s_1 \le s_2$ implies $q_1 \le q_2$ and $\l I_1'(q_1) = 0$. The second inequality holds by~\eqref{argmax1}.

By Lemma~\ref{lemma:envelope}, $\CC[I_1(q)-I_2(q)]=I_1(q)-I_2(q)$ for $q \in [0,q_2]$ if and only if $\a \ge g_1(q_2)=a_{12}$ (also see Figure~\ref{fig:ap0} (a)). 
\item[2.]
We first show that $\Qc = [q_1,1]$ satisfies ~\eqref{eff-pbc} with $I_1(\cdot)$ and $I_2(\cdot)$ interchanged. Specifically we show that 
$I_1(q)+ \mathfrak{C}[ \lambda I_2(q)- I_1(q)]$ is increasing in $[0,q_1]$. Since $I_1(q)+ \mathfrak{C}[ \lambda I_2(q)- I_1(q)]$ is concave, it suffices to show that the derivative of $I_1(q)+ \mathfrak{C}[ \lambda I_2(q)- I_1(q)]$ at $q=q_1$ is nonnegative. Consider

\begin{align*}
\frac{d \big(I_1(q)+ \CC[\l I_2(q) - I_1(q)]\big)}{dq} \bigg |_{q=q_1}
 &=\frac{d \big(I_1(q)+ \l \CC[I_2(q) - I_1(q)/\l]\big)}{dq} \bigg |_{q=q_1}\\
 &\stackrel{(a)}{=}\begin{cases}
 \l I_2'(q_1) &\text{ if } q_1 \ge r'\\
I_1'(q_1) + (\l I_2(r') - I_1(r'))/r' &\text{ otherwise.}
\end{cases} \\
&\stackrel{(b)}{\ge} 0
\end{align*}
for $r'=r(\a')$ where $\a' = \a/\l$. Step $(a)$ follows by rewritting $I_1(q)/\l$ for $A_1 = \a$ as $I_1(q)$ for $A_1 = \a/\l$ (see~\eqref{I_i}) and then applying Lemma~\ref{lemma:envelope}. Step $(b)$ holds because $\l I_2'(q_1) \ge 0$ and $(\l I_2(r')-I_1(r'))/r' \ge 0$. The first inequality holds since $s_1 \le s_2$ implies $q_1 \le q_2$ and $\l I_2'(q_2) = 0$. The second inequality holds by~\eqref{argmax1}.

By Lemma~\ref{lemma:envelope}, $\CC[I_2(q)-I_1(q)]=I_2(q)-I_1(q)$ for $q \in [q_1,1]$ if and only if $\a \le g_2(q_1)=a_{23}$ (also see Figure~\ref{fig:ap0} (b)). 
%
%
%
%
%
\end{enumerate}
\end{proof}

Figure~\ref{fig:sp_optimal} illustrates the parameter ranges for which superposition coding is optimal. As can be seen the area in the $\a$-$s_2$ plane where superposition coding is not optimal becomes smaller as $s_1$ increases. In Appendix~\ref{volume}, we show that the fraction of the channel parameter space for which superposition coding is optimal approaches one; hence superposition coding is in a sense almost always optimal for the P-BC. In comparison, the fraction of the parameter space for which the P-BC is degraded is always bounded away from 1.\\

\begin{figure}[htpb]
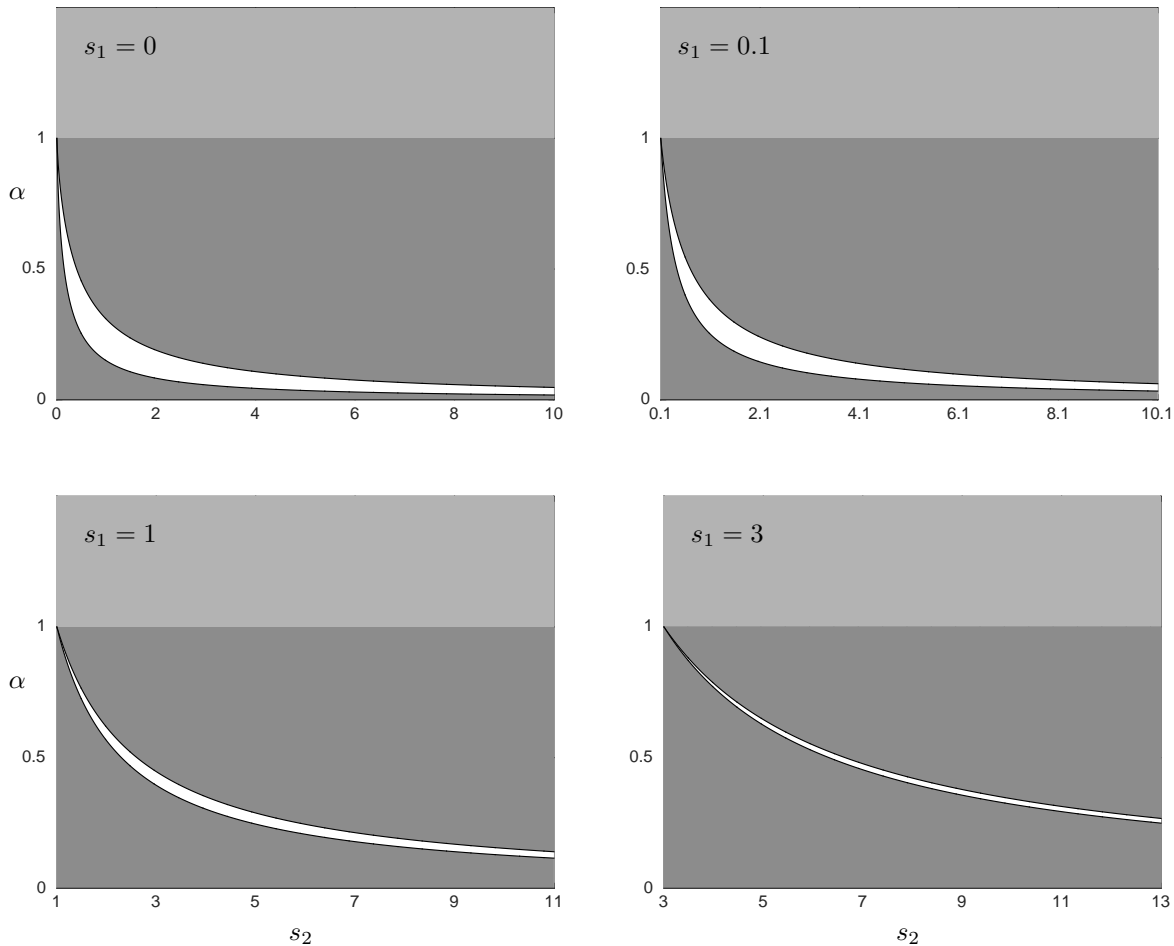

\begin{center}
\begin{tabular}{ccc}
\psfrag{p}[c]{}
\psfrag{s}[t]{$s_2$}
\psfrag{a}[c]{$\a$}
\psfrag{k1}[t]{$\ \ \ \ \ s_1=0$}
\includegraphics[scale=0.42]{optimal_1}&&
\psfrag{p}[c]{}
\psfrag{a}[c]{$\a$}
\psfrag{s}[t]{$s_2$}
\psfrag{k2}[t]{$\ \ \ \ \ s_1=0.1$}
\includegraphics[scale=0.42]{optimal_2}\\&&\\
&&\\
\psfrag{p}[c]{}
\psfrag{a}[c]{$\a$}
\psfrag{s}[t]{$s_2$}
\psfrag{k3}[t]{$\ \ \ \ \ s_1=1$}
\includegraphics[scale=0.42]{optimal_3}&&
\psfrag{p}[c]{}
\psfrag{a}[c]{$\a$}
\psfrag{s}[t]{$s_2$}
\psfrag{k4}[t]{$\ \ \ \ \ s_1=3$}
\includegraphics[scale=0.42]{optimal_4}\\&&\\
\end{tabular}
\caption{Plots of $\a$ versus $s_2$ for $s_1=0$ (top left), $s_1=0.1$ (top right), $s_1=1$ (bottom left) and $s_1=3$ (bottom right). The shaded areas in each plot are where superposition coding is optimal and the light shaded area ($\a\ge 1$) is where $Y_2$ is a degraded version of $Y_1$.}
\label{fig:sp_optimal}
\end{center}
\end{figure}

As illustrated in Figure~\ref{fig:line}, the capacity region of the P-BC is achieved using superposition coding in the ranges $\a \leq \a_{23}$ and $\a \ge \a_2$. Note that for the remaining range of $\a$ the channel is not more capable or effectively less noisy, which follows immediately from the if and only if conditions established in Theorems~\ref{thm:lnmc} and~\ref{thm:eff}. In the following section we explore bounds on the capacity region in the range for $\a \in (\a_{23}, \a_2)$.

\section{Gap between Marton and UV}\label{sec:gap}
The best known inner and outer bounds on the capacity region of the DM-BC are the Marton inner bound~\cite{Marton1979} and the UV outer bound~\cite{Nair--EG2007}, respectively. We show that for $\a \in (\a_{23}, \a_2)$ there can be a gap between these two bounds.

In~\cite{Geng--Jog--Nair--Wang2013}, Geng, Jog, Nair and Wang showed that for binary input broadcast channels, Marton's inner reduces to the set of rate pairs $(R_1,R_2)$ such that
\begin{align}\begin{split}\label{marton-inner}
R_1 &< I(W;Y_1) + \sum_{j=1}^k\beta_j I(X;Y_1|W=j),\\
R_2 &< I(W;Y_2) + \sum_{j=k+1}^5 \beta_j I(X;Y_2|W=j),\\
R_1 + R_2 &< \min\{I(W;Y_1), I(W;Y_2)\} +  \sum_{j=1}^k\beta_j I(X;Y_1|W=j) + \sum_{j=k+1}^5 \beta_j I(X;Y_2|W=j)
\end{split}\end{align}
for some $p_W(j)=\beta_j$, $j\in [1:5]$, and $p(x|w)$. This region is achieved using \emph{randomized time-division}~\cite{Hajek--Pursley1979}. 
This ingenious insight helps simplify the computation of Marton's inner bound for the binary P-BC; hence for the P-BC itself. 

The UV outer bound on the capacity region of the DM-BC is the set of rate pairs $(R_1,R_2)$ such that
\begin{align}\begin{split}\label{uv-outer}
R_1 & \le I(V;Y_1), \\
R_2 & \le I(U;Y_2), \\
R_1 + R_2 & \le I(V;Y_1) + I(U;Y_2|V),\\
R_1 + R_2 & \le I(V;Y_1|U) + I(U;Y_2)
\end{split}\end{align}
for some $p(u, v)$ and function $x(u,v)$, $|\Uc|, |\Vc| \le |\Xc| + 1$. Computing this bound even for binary input broadcast channels is quite difficult. Hence, we compute the maximum sum rates for the Marton and the UV bounds instead of the complete bounds.  
Figure~\ref{fig:gap} plots the maximum sum rates for $\a_{23} \leq \a \leq \a_2$ when $s_1=0.1$ and $s_2=1$ ($\a_{23}=0.27, \a_2=0.4$). Note that for $0.27 \le \a \le 0.286$, the sum rates coincide. For the rest of the range there is a small gap between the Marton and the UV bound sum rates. In particular for $\a=0.34$, the gap is approximately $0.0039$.\\

\begin{figure}[htbp]
\begin{center}
\psfrag{a}[bc][BC]{$\a$\, }
\psfrag{0}[bc][BC]{$0$}
\psfrag{1}[bc][BC]{$1$\ }
\psfrag{d}[r][cc]{} 
\psfrag{n}[r][cc]{}
\psfrag{m}[r][cc]{}
\psfrag{e}[r][cc]{}
\psfrag{a1}[cc][cc]{$\a_4$\ }
\psfrag{a2}[cc][cc]{$\a_3$\ }
\psfrag{a3}[cc][cc]{{\highlight $\a_2$}\ }
\psfrag{a34}[cc][cc]{$\a_{12}$\ }
\psfrag{a12}[cc][cc]{{\highlight $\a_{23}$}} 
\psfrag{a4}[cc][cc]{$\a_1$\ }
\psfrag{k}[c]{}
\psfrag{n2}[l][cc]{}
\psfrag{m2}[l][cc]{}
\psfrag{y}[b]{Max. sum rate}
\psfrag{a}[t]{$\a$}
\psfrag{M}[l][cc]{{\footnotesize Marton}}
\psfrag{U}[l][cc]{{\footnotesize UV}}
\psfrag{C}[l][cc]{{\footnotesize Superposition}}
\includegraphics[scale=0.5]{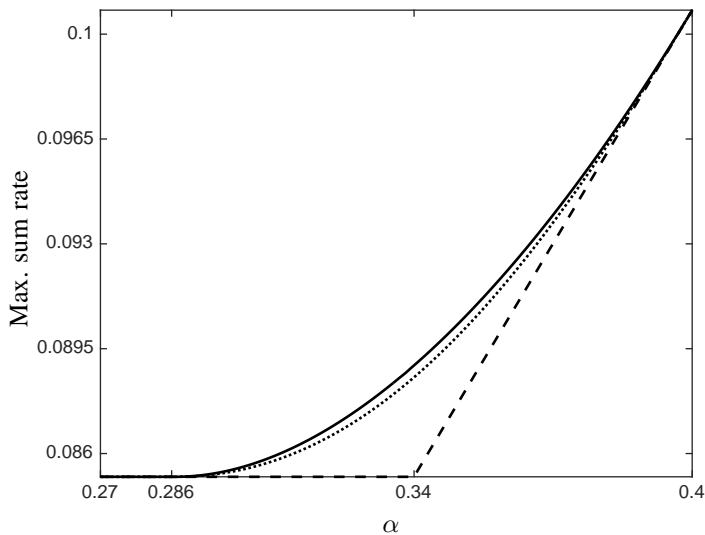}
\end{center}
\caption{Plots of the maximum Marton sum rate (dot), superposition sum rate (dash), and UV sum rate for $s_1 = 0.1, s_2 = 1$, and $\a \in [\a_{23}, \a_2]$.}
\label{fig:gap}
\end{figure}
It turns out that superposition coding is optimal in the range where Marton's sum rate and the $UV$ outer bound sum rate in Figure~\ref{fig:gap} coincide, i.e., for  $\a \in [0.27, 0.286]$. To show this, note that the condition for effectively less noisy in~\eqref{setPP} can be tightened further. It is clearly sufficient that $I(U;Y_1) \ge I(U;Y_2)$ for every $p^*(u,x) \in \Pr(U,X)$ and every $0 \le \l \le 1$ such that 
\begin{align}\label{stronger}
\max_{p(u,x)} \big(\lambda I(X;Y_1|U) + I(U;Y_2)\big) = \max_{p^*(u,x) \in \Pr(U,X)} \big(\lambda I(X;Y_1|U) + I(U;Y_2)\big). 
\end{align}
This is the condition satisfied for the P-BC in the range $\a \in [0.27,0.286]$. However, unlike the looser condition~\eqref{setPP}, where we are able to express the class of pmfs in terms of the concave envelope of $I(X;Y_1) - I(X;Y_2)$ and determine analytically the range in which the P-BC is effectively less noisy, we can only numerically evaluate the above condition.


\section{Average power constraint}\label{sec:extensions}
The results on superposition coding in sections~\ref{sec:superposition}-~\ref{sec:eff} can be readily extended to the case when there is also an average power constraint. 
Suppose that in addition to the maximum power constraint $X(t) \le 1$ (which is needed for the capacity to be finite) there is an average power constraint , i.e.,
\begin{align}\label{averagePower}
\frac{1}{T}\int_0^T X(t) dt \le \sigma.
\end{align}

The capacity region of P-BC for the average power constraint setting is equal to the capacity region of the corresponding binary P-BC for $\E[X] \le \sigma$. This is a simple extension for Wyner's argument that the capacity of point-to-point Poisson channel is the equal to the capacity of the corresponding binary channel for $\E[X] \le \sigma$~\cite{Wyner1988a, Wyner1988b}.


Note that if a broadcast channel is less noisy or more capable, the UV outer bound in~\eqref{uv-outer} with an input constraint coincides with the superposition inner bound with the same input constraint. Hence, if a P-BC is less noisy, it is also less noisy under the average power constraint, and similarly if a P-BC is more capable, it is also more capable under the average power constraint. 

In contrast, if a P-BC is effectively less noisy, it is not necessarily effectively less noisy under the average power constraint. The results of effectively less noisy, however, can be easily extended to the average power constraint setting.

\begin{theorem}\textnormal{
For P-BC with an average power constraint in~\eqref{averagePower},
\begin{enumerate}
\item[1.] Receiver $Y_1$ is effectively less noisy than $Y_2$ if $\a \ge g_1(\min\{\sigma,q_2\})$, 
\item[2.] Receiver $Y_2$ is effectively less noisy than $Y_1$ if $\a \le g_2(\min\{\sigma,q_1\})$ 
\end{enumerate}
where the functions $g_1(\cdot)$ and $g_2(\cdot)$ are defined in~\eqref{g1} and~\eqref{g2}, respectively.
}\end{theorem}

\begin{proof}

Note that $X \sim \Bern(q)$ for $q \in [0,\sigma]$ satisfies the average power constraint. The condition in Definition~\ref{def:effectivelyP-BC} can be modified to the average power constraint setting.
Let $\Qc \subseteq [0,\sigma]$ be such that for every $0 \le \lambda \le 1$,
\begin{align}\label{averageQ}
\max_{q\in[0,\sigma]} I_2(q) + \CC[\l I_1(q) - I_2(q)] = \max_{q\in\Qc} I_2(q) + \CC[\l I_1(q) - I_2(q)].
\end{align}
Receiver $Y_1$ is effectively less noisy than $Y_2$ under the average power constraint if $I_1(q) - I_2(q) = \CC[I_1(q) - I_2(q)]$ for every $q \in \Qc$.

\begin{enumerate}
\item[1.] Recall that for P-BC without an average power constraint, $Y_1$ is effectively less noisy than $Y_2$ if 
\[
I_1(q)-I_2(q) = \CC[I_1(q)-I_2(q)] \text{ for } q \in [0,q_2].
\]

For P-BC with an average power constraint $q \le \sigma$, $Y_1$ is effectively less noisy than $Y_2$ if
\[
I_1(q)-I_2(q) = \CC[I_1(q)-I_2(q)] \text{ for } q \in [0,\min\{\sigma,q_2\}].
\]
If $\a \ge g_1(\min\{\sigma,q_2\})$, $Y_1$ is effectively less noisy than $Y_2$ under the average power constraint.
\item[2.]
Recall that for P-BC without an average power constraint, $Y_2$ is effectively less noisy than $Y_1$ if 
\[
I_2(q)-I_1(q) = \CC[I_2(q)-I_1(q)] \text{ for } q \in [q_1,1].
\]
If there is an average power constraint $q \le \sigma$ for some $\sigma \ge q_1$, then $Y_2$ is effectively less noisy than $Y_1$ under the average power constraint if
\[
I_2(q)-I_1(q) = \CC[I_2(q)-I_1(q)] \text{ for } q \in [q_1,\sigma].
\]
Thus for $\a \le g_2(q_1)$, $Y_2$ is effectively less noisy than $Y_1$ under the average power constraint.

If there is an average power constraint $q \le \sigma$ for some $\sigma < q_1$, then $Y_2$ is effectively less noisy under the average power constraint than $Y_1$ if 
\[
I_2(q)-I_1(q) = \CC[I_2(q)-I_1(q)] \text{ for } q = \sigma.
\]
Thus for $\a \le g_2(\sigma)$, $Y_2$ is effectively less noisy under the average power constraint than $Y_1$.
\end{enumerate}
\end{proof}

\section{Final remarks}\label{sec:remarks}
We showed that superposition coding is optimal for almost all Poisson broadcast channels, and that when superposition is not optimal, there is a gap between Marton's inner bound and the UV outer bound. Hence the capacity region for the P-BC is still not known in general. 

We note that in~\cite{Geng--Nair--Shamai--Wang2013}, Geng, Nair, Shamai, and Wang similarly showed that for the class of binary input symmetric output broadcast channels (which do not include the binary P-BC) either superposition coding is optimal or there is a gap between Marton's inner bound and the UV outer bound. 

We introduced the effectively less noisy broadcast channels for which superposition coding is optimal. This condition for effectively less noisy can be further extended, but we can only verify it numerically.


Why is this the case and does it hold for general binary input broadcast channels?

The intuitive reason superposition coding is almost always optimal for the P-BC is that the binary P-BC is more capable (and even less noisy) for most parameter ranges as established in Appendix~\ref{volume}. 
 Hence one channel is almost always stronger than the other.

%

For other classes of binary input broadcast channels, the more capable condition is much less likely to be satisfied.
As an extreme case, consider the skewed binary broadcast channel in Figure~\ref{bssc}-(a), which is a generalization of the skew symmetric BC in~\cite{Hajek--Pursley1979}. We can show 
that 
 the channel is not more capable (hence also not less noisy or degraded) for every $p_1, p_2 \in (0,1)$. It turns out, however, that for $(p_1,p_2)$ in the dark shaded area in Figure~\ref{bssc}-(b), the channel is  effectively less noisy, and for $(p_1,p_2)$ in the lightly shaded area in Figure~\ref{bssc}-(b), the BC is not effectively less noisy but superposition coding is still optimal (which is shown by verifying the stronger condition in~\eqref{stronger}). These shaded areas constitute $76\%$ of the parameter space area! This clearly demonstrates that our intuition about when superposition coding  is optimal does not always hold. The unshaded area in~Figure~\ref{bssc}-(b) is where Marton's sum rate is strictly greater than the superposition sum rate (i.e., $\max\{C_1,C_2\}$), which implies that superposition coding is not optimal. 
 

\begin{figure}[h]
\begin{center}
\begin{tabular}{ccccc}
\psfrag{0}[r]{$0$}
\psfrag{1}[r]{$1$}
\psfrag{z}[l]{$1$}
\psfrag{o}[l]{$0$}
\psfrag{z1}[l]{$0$}
\psfrag{o1}[l]{$1$}
\psfrag{q}[b]{$a_1$} 
\psfrag{p'}[t]{$b_1$} 
\psfrag{p}[b]{$a_2$} 
\psfrag{q'}[t]{$b_2$} 
\psfrag{X}[c]{$X$}
\psfrag{Y1}[c]{$Y_1$}
\psfrag{Y2}[c]{$Y_2$}
\psfrag{q}[b]{$p_1$} 
\psfrag{z}[l]{$0$}
\psfrag{o}[l]{$1$}
\psfrag{q'}[t]{$p_2$} 
\psfrag{a}[c]{}
\includegraphics[scale=0.55]{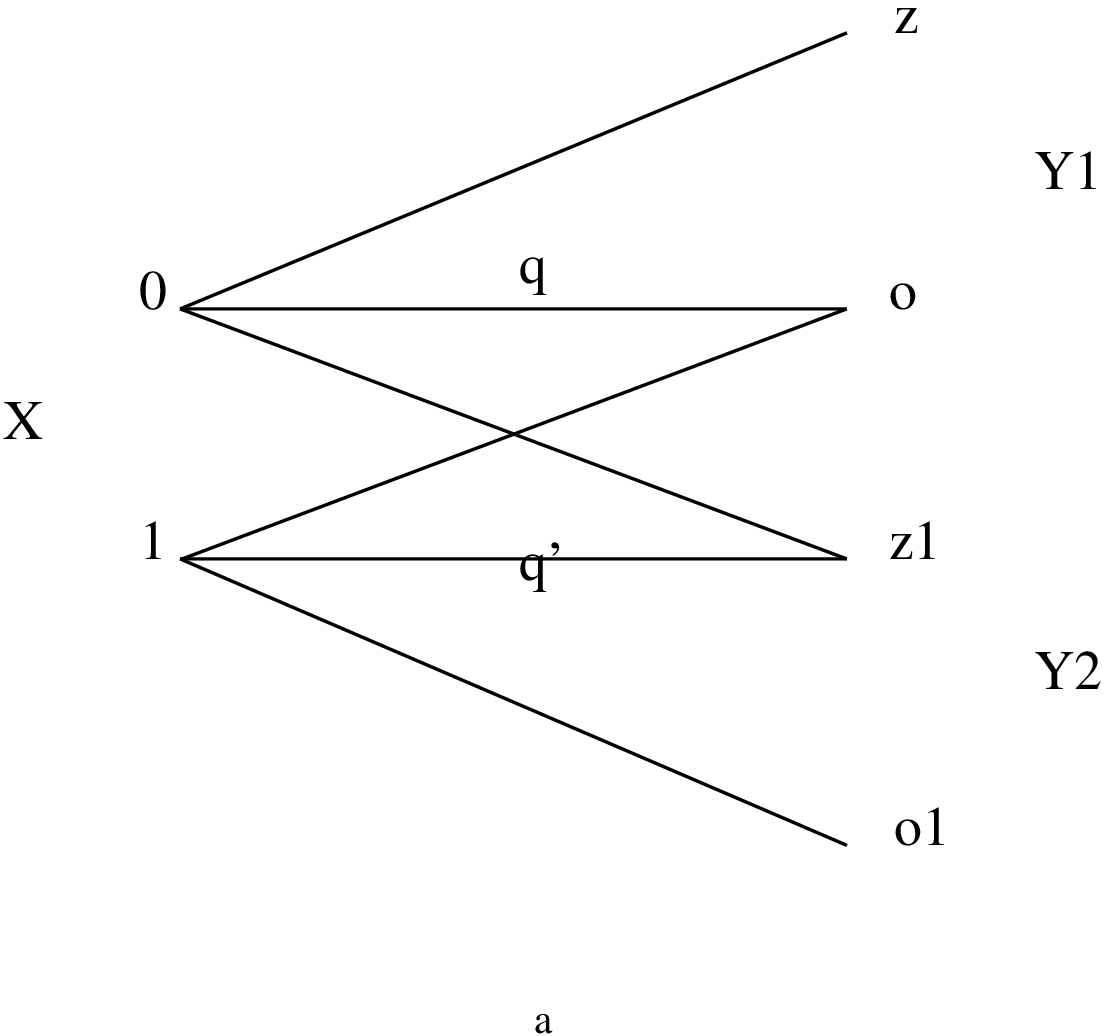}&&&&
\psfrag{A}[c]{\ \ Eff. less noisy}
\psfrag{B}[c]{SC}
\psfrag{C}[c]{SC}
\psfrag{D}[c]{suboptimal}
\psfrag{E}[c]{optimal}
\psfrag{p}[c]{}
\psfrag{p1}[c]{$p_1$}
\psfrag{p2}[c]{$\ \ p_2$}
\includegraphics[scale=0.45]{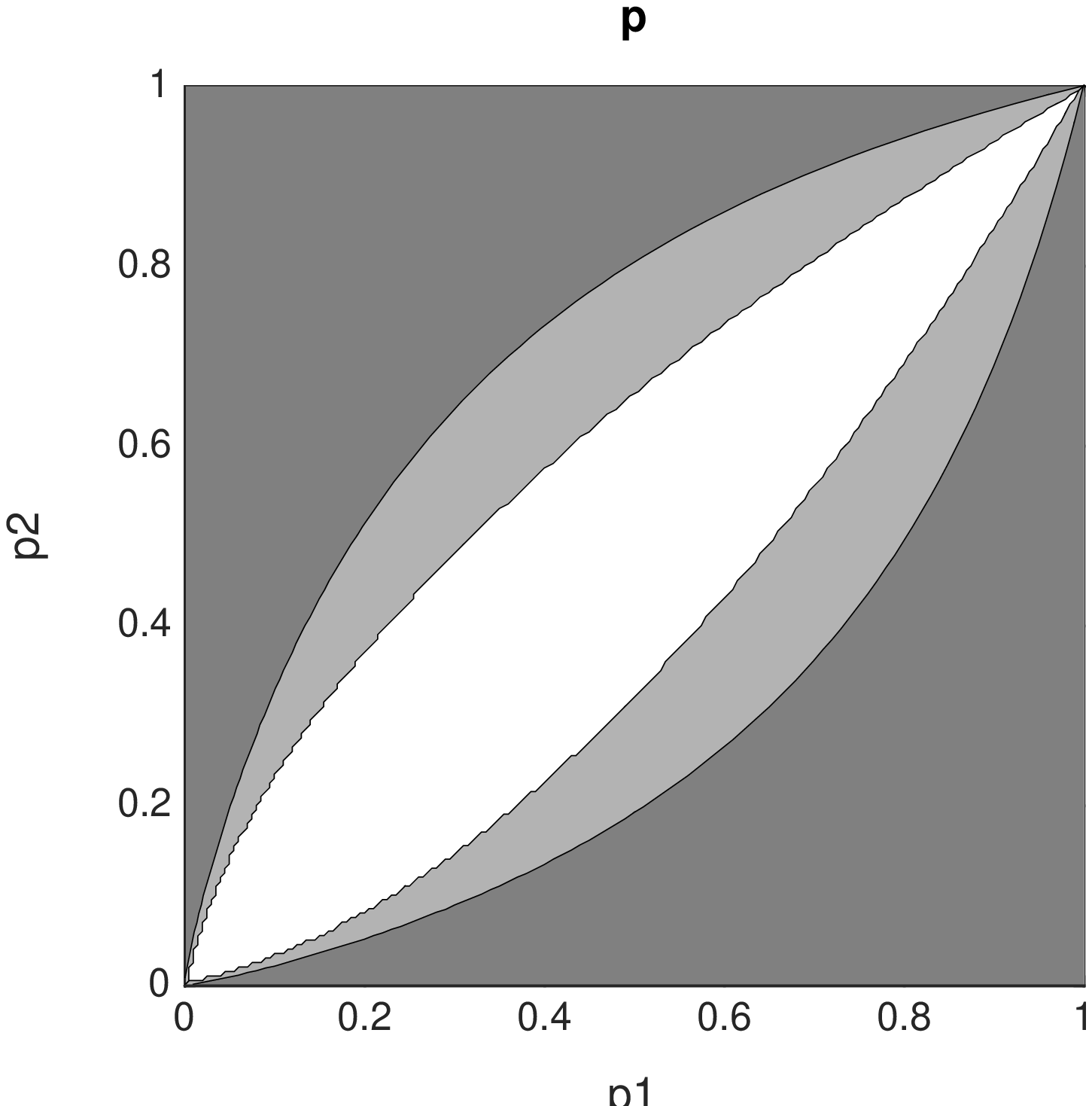}\\&&&&\\
(a) &&&& (b)
\end{tabular}
\caption{(a) skewed binary broadcast channel (b) Plot of $p_1$ vs $p_2$. The lightly shaded area is where superposition coding is optimal and the dark shaded area is where one receiver is effectively less noisy than the other receiver.}
\label{bssc}
\end{center}
\end{figure}

To perform the analytical and computational evaluations, we relied heavily on the concave envelope method that has been used in other applications, including~\cite{Geng--Nair2014} which establishes the optimality of dirty paper coding and superposition coding for MIMO broadcast channels with common message, and~\cite{Geng--Gohari--Nair--Yu2014} which shows that UV outer bound is not tight.

Finally, it would be interesting to find a similar extension of the more capable to the notion of effectively less noisy. The difficulty is finding an outer bound similar to the one we used for effectively less noisy. 

%
%
%

\bibliographystyle{IEEEtran}
\bibliography{nit}
\appendices
\section{Proof of Lemma~\ref{lemma:envelope}}\label{proof:envelope}
We first prove that  the function $(I_2(q) - I_1(q))$ can change concavity at most once in the range $q\in[0,1]$. In particular, the function $I_2(q) - I_1(q)$ is (i) concave in $[0,1]$ if $\a\le\a_4$, (ii) convex in $[0,1]$ if $\a\ge\a_1$, or (iii) convex in $[0,\kappa]$ and concave in $[\kappa,1]$ otherwise, where
\begin{align}\label{kappa}
\kappa = \frac{\a s_2 -s_1}{1-\a} \in [0,1].
\end{align}
To show the above statement, note that the second derivative of $I_2(q) - I_1(q)$,
\begin{align}\label{eq:diff}
I_2''(q) - I_1''(q)=\frac{(\a-1)q + \a s_2-s_1}{(q+s_1)(q+s_2)},
\end{align}
has a unique zero, is nonnegative if $\a \ge \a_1$, and is nonpositive if $\a \le \a_4$. Now we are ready to prove Lemma~\ref{lemma:envelope}.

\begin{enumerate}
\item[1.] 
Let $f_1(q) = (I_1(q) - I_2(q))/(1-q)$. We first show that $t$ defined in~\eqref{def:t} satisfies
\begin{align}\label{argmax1}
t = \arg\max_{q\in[0,1]}f_1(q). 
\end{align}
It can be shown that
\begin{align*}
f_1'(q) = (1-q)^{-2} (\a - g_1(q)) \bigg((1+s_1)\log\frac{1+s_1}{q+s_1}-1+q\bigg),
\end{align*}
where $g_1(q)$ is defined in~\eqref{g1}. It can be easily checked that $f_1'(q)$ and $(\a-g_1(q))$ have the same sign for $0<q<1$ and $g_1(q)$ is increasing for $q>0$. Hence

\begin{enumerate}\itemsep 0.2em
\item
if $\a \le g_1(0)=\a_3$, then $f_1'(q) \le 0$ in $[0,1]$ and $f_1(q)$ is maximized at $t$.
\item
if $\a \ge \lim_{q\to 1}g_1(q)=\a_1$, then $f_1'(q) \ge 0$ in $[0,1]$ and $f_1(q)$ is maximized at $t$.
\item
otherwise, $f_1'(q) \ge 0$ in $[0,g_1^{-1}(\a)]$ and $f'(q) \le 0$ in $[g_1^{-1}(\a),1]$ and $f_1(q)$ is maximized at $t$.
\end{enumerate}

Let 
\begin{align*}
&L_1(q)=\begin{cases}
I_1(q)-I_2(q) &\text{ if } 0\leq q \leq t,\\
(1-q)(I_1(t)-I_2(t))/(1-t)&\text{ if } q > t.
\end{cases}
\end{align*}
Then $L_1(q) \ge I_1(q)-I_2(q)$ because $f_1(t) \ge f_1(q)$ for any $q \in [0,1]$.

On the other hand, note that for $q > t$, $L_1(q)$ can be expressed as 
\[
L_1(q) = \big((1-q)/(1-t)\big)(I_1(t)-I_2(t)) +  \big((q-t)/(1-t)\big)(I_1(1)-I_2(1)).
\]
Since $q = \big((1-q)/(1-t)\big)\cdot t +  \big((q-t)/(1-t)\big) \cdot 1$, it follows that for $q >t$, $L_1(q)$ is a convex combination of $I_1(t)-I_2(t)$ and $I_1(1)-I_2(1)$. Thus $\mathfrak{C}[I_1(q)-I_2(q)] \ge L_1(q)$. 

Finally to argue that $L_1(q) = \CC[I_1(q) - I_2(q)]$, we are left to show that $L_1(q)$ is indeed concave.
We first show that $L_1(q)$ is concave in $[0,t]$. Recall that $I_1(q)-I_2(q)$ is concave in $[0, \kappa]$ for $\kappa$ in~\eqref{kappa}. Thus it suffices to show that $t \le \kappa$.  Equivalently we show that  $f_1'(\kappa) \le 0$. Consider 
\begin{align*}
  f_1'(q) &=\frac{I_1'(q)-I_2'(q)}{1-q} - \frac{(I_1(1)-I_2(1))-(I_1(q) - I_2(q))}{(1-q)^2}\\
	  &= (I_1'(q)-I_2'(q))/(1-q) - (I_1'(u) - I_2'(u))/(1-q)
\end{align*}
for some $q \le u\le 1$. The inequality $f_1'(\kappa) \le 0$ holds since $I_1''(q)-I_2''(q) \ge 0$ for $q \ge \kappa$ as can be seen from~\eqref{eq:diff}. 
 We conclude that $L_1(q)$ is concave because it is concave in $[0,t]$ and $[t,1]$ and is differentiable at $t$.

\item[2.] 
Let $f_2(q)=(I_2(q)-I_1(q))/q$. We show that $r$ defined in~\eqref{def:r} satisfies
\begin{align}
r = \arg\max_{q\in[0,1]}f_2(q).\label{argmax2}
\end{align}
It can be shown that
\begin{align*}
f_2'(q) = q^{-2}(\a - g_2(q))(q -s_1\log(1+q/s_1)),
\end{align*}
where $g_2(q)$ is defined in~\eqref{g2}. It can be easily seen that $f_2'(q)$ and $(\a-g_2(q))$ have the same sign for $0<q<1$. Also it can be shown that $g_2(q)$ is increasing for $q>0$. Hence

\begin{enumerate}\itemsep 0.2em
\item
if $\a \le \lim_{q \to 0}g_2(q)=\a_4$, then $f_2'(q) \le 0$ in $[0,1]$ and $f_2(q)$ is maximized at $r$.
\item
if $\a \ge g_2(1)=\a_2$, then $f_2'(q) \ge 0$ in $[0,1]$ and $f_2(q)$ is maximized at $r$.
\item
otherwise, $f_2'(q) \ge 0$ in $[0,g_2^{-1}(\a)]$ and $f_2'(q) \le 0$ in $[g_2^{-1}(\a),1]$ and $f_2(q)$ is maximized at $r$.
\end{enumerate}
Let
\begin{align*} 
L_2(q)&=\begin{cases}
(q/r)(I_2(r)-I_1(r))&\text{ if } 0\leq q < r,\\
I_2(q)-I_1(q) &\text{ if } q \geq r.
\end{cases}
\end{align*} 
Then $I_2(q)-I_1(q) \leq L_2(q)$ because $f_2(r) \ge f_2(q)$ for any $q \in [0,1]$.

On the other hand, $\mathfrak{C}[I_2(q)-I_1(q)] \ge L_2(q)$ because for $0 \leq q \leq r$, $L_2(q)$ is a convex combination of $I_2(r)-I_1(r)$ and $I_2(0)-I_1(0)$. To complete the proof we now show that $L_2(q)$ is concave. Recall that $I_2(q)-I_1(q)$ is concave in $[\kappa,1]$ for $\kappa$ in~\eqref{kappa}. To show that $L_2(q)$ is concave in $[r,1]$, we show $\kappa \le r$, i.e. $f_2'(\kappa) \ge 0$. Consider 
\begin{align*}
  f_2'(q) &= \frac{I_2'(q)-I_1'(q)}{q} - \frac{I_2(q) - I_1(q)}{q^2}\\
	  &= (I_2'(q)-I_1'(q))/q - (I_2'(u) - I_1'(u))/q
\end{align*}
for some $u \leq q$. Since $I_2''(q)-I_1''(q) \ge 0$ in $[0,\kappa]$, it follows that $f_2'(\kappa)\ge 0$. We conclude that $L_2(q)$ is concave because it is concave in $[0,r]$ and $[r,1]$ and is differentiable at $r$.
\end{enumerate}


\section{Order of breakpoints}\label{order}
Since less noisy channels are also more capable it follows that $\a_4 \le \a_3$ and $\a_2 \le \a_1$. Also $Y_1$ and $Y_2$ cannot be more capable simultaneously unless the two channels are identical. Thus, $0 \le \a_4 \le \a_3 \le \a_2 \le \a_1$. To completely characterize the order, we need to show the following:
\[
\a_2 \stackrel{(a)}{\le} \a_{12} \stackrel{(b)}{\le}\a_1 \text{ and }\a_3 \stackrel{(c)}{\le} \a_{23} \stackrel{(d)}{\le} \a_2.
\]
Note that inequalities $(b)$ and $(d)$ follow by Lemma~\ref{lemma:envelope}. The remaining inequalities were shown through simulation. To show $(a)$, we checked using Mathematica that the minimum of $\a_{12} - \a_2$ for $0\le s_1\le s_2$ is nonnegative. To show $(c)$, let
\[
w(x) = \frac{x\log(1+1/x)-1}{(1+x)\log(1+x)-(1+x)\log(q_2+x)-1+q_2}.
\]
Then the condition $\a_3 \le \a_{23}$ is equivalent to $w(x) \le w(s_2)$ for all $x \le s_2$ for any  $q_2$, and we checked that this is true using Mathematica and Maple.

%
%

\section{volume}\label{volume}
Consider the set of channel parameters:  $(\a, s_1, s_2) \in [0,1] \times [0,b] \times [0, kb]$ for $b,k \geq 0$.
Let $b_0=\min\{b,kb\}$. The fraction of the set of P-BC parameters for which the channel is degraded is 
\begin{align*}
\frac{1}{kb^2}\iint_{s_1 \ge s_2} ds_2\, ds_1 =\begin{cases}
0.5k^{-1} &\text{ if }k \ge 1,\\
1-0.5k &\text{ if } k<1.
\end{cases}
\end{align*}
The fraction of the set of P-BC parameters for which superposition coding is optimal is lower bounded by the fraction for which the P-BC is less noisy. 
Consider 
{\allowdisplaybreaks
\begin{align*}
&1 - \frac{1}{kb^2}\iint_{s_1 \le s_2}  \int_{s_1/s_2}^{(1+s_1)/(1+s_2)} d\alpha \, ds_2\, ds_1\\
&=1 - \frac{1}{kb^2}\int_{0}^{{b_0}}\int_{s_1}^{kb}\int_{s_1/s_2}^{(1+s_1)/(1+s_2)}  d\alpha\, ds_2\, ds_1\\
&=1-\frac{1}{kb^2}\int_{0}^{{b_0}}\int_{s_1}^{kb}\bigg(\frac{1+s_1}{1+s_2}-\frac{s_1}{s_2}\bigg) ds_2\, ds_1\\
&=1-\frac{1}{kb^2}\int_{0}^{{b_0}}\bigg((1+s_1)\log\frac{1+kb}{1+s_1}-s_1\log\frac{kb}{s_1}\bigg) ds_1\\
&=1-\frac{1}{kb^2}\bigg(\frac{(1+{b_0})^2-1}{2}\log (1+kb) - \frac{{b_0}^2}{2}\log kb - \frac{1+{b_0}^2}{2} \log (1+{b_0}) + \frac{{b_0}^2\log{b_0}}{2} + \frac{b_0}{2}\bigg)\\
&=1\!-\!\frac{1}{kb^2}\bigg(\frac{{b_0}^2}{2}\big(\log\frac{1+kb}{kb}\!-\!\log\frac{1+{b_0}}{{b_0}}\big)-\frac{\log(1+{b_0})}{2}+{b_0}\big(\log(1+kb)-\log(1+{b_0})+\frac{1}{2}\big)\bigg),
\end{align*}
which approaches to $1$ as $b \to \infty$.


\end{document}